\documentclass[12pt,american,nofootinbib]{revtex4-1}
\usepackage{lmodern}

\usepackage[T1]{fontenc}
\usepackage[utf8]{inputenc}
\usepackage{geometry}
\usepackage{color}
\usepackage{ulem}
\usepackage{babel}
\usepackage{enumitem}
\usepackage{amsmath}
\usepackage{amsthm}
\usepackage{amssymb}
\usepackage[unicode=true,pdfusetitle,
 bookmarks=true,bookmarksnumbered=false,bookmarksopen=false,
 breaklinks=false,pdfborder={0 0 0},pdfborderstyle={},
 backref=false,colorlinks=false]
 {hyperref}

\makeatletter
\theoremstyle{plain}
\newtheorem{thm}{\protect\theoremname}
\theoremstyle{plain}
\newtheorem{prop}{\protect\propositionname}

\makeatother

\providecommand{\propositionname}{Proposition}
\providecommand{\theoremname}{Theorem}

\begin{document}
\title{Relativistic cosmology and intrinsic spin of matter: Results
and theorems in Einstein-Cartan theory}
\author{
\vskip -0.5cm Paulo Luz}
\email{paulo.luz@ist.utl.pt}
\affiliation{Centro de Astrof\'isica e Gravita\c c\~ao -
CENTRA, 
Instituto Superior T\'ecnico - IST, Universidade de Lisboa - UL,
Avenida Rovisco
Pais 1, 1049-001 Lisboa}
\affiliation{Departamento de Matem\'atica, ISCTE - Instituto
Universit\'ario de Lisboa, Avenida das For\c cas Armadas,
1649-026 Lisboa, Portugal}

\author{
\vskip -0.3cm Jos\'e P. S. Lemos}
\email{joselemos@ist.utl.pt}
\affiliation{Centro de Astrof\'isica e Gravita\c c\~ao -
CENTRA, Departamento de F\'isica,
Instituto Superior T\'ecnico - IST, Universidade de Lisboa - UL,
Avenida Rovisco
Pais 1, 1049-001 Lisboa}

\begin{abstract}

We start by presenting the general set of structure equations for the
1+3 threading spacetime decomposition in 4 spacetime dimensions, valid
for any theory of gravitation based on a metric compatible affine
connection. We then apply these equations to the study of cosmological
solutions of the Einstein-Cartan theory in which the matter is
modeled by a perfect fluid with intrinsic spin.  It is shown that the
metric tensor can be described by a generic FLRW solution.  However,
due to the presence of torsion the Weyl tensors might not vanish.  The
coupling between the torsion and Weyl tensors leads to the conclusion
that, in this cosmological model, the universe must either be flat or
open, excluding definitely the possibility of a closed universe.  In
the open case, we derive a wave equation for the traceless part of the
magnetic part of the Weyl tensor and show how the intrinsic spin of
matter in a dynamic universe leads to the generation and emission of
gravitational waves. Lastly, in this cosmological model, it is found
that the torsion tensor, which has an intrinsic spin as its source,
contributes to a positive accelerated expansion of the
universe. Comparing the theoretical predictions of the model with the
current experimental data, we conclude that torsion cannot
completely replace the role of a cosmological constant.

\end{abstract}

\maketitle

\section{Introduction}

\subsection{The torsion tensor field and the Einstein-Cartan theory }

A Riemannian spacetime geometry is uniquely described by the metric
tensor field, in that the Riemann curvature tensor is solely given by
the metric and its first and second derivatives.  Moreover, the affine
connection, which is the structure that defines parallel transport
between tensors, is metric compatible and symmetric in Riemannian
geometry, and called Levi-Civita connection.  The geodesic equation,
which determines the shortest and longest curves between two
infinitesimally close points, is in Riemannian geometry
also an equation
for curves that transport tangent vectors in a autoparallel
manner. General relativity, a spacetime theory of gravitation, has as
one of its intrinsic assumptions the fact that the underlying geometry
is Riemannian.  In general relativity, the link between geometry and
matter is provided by Einstein equation, that equates the Einstein
tensor, which is a contraction of the Riemann tensor, to the matter
stress-energy tensor.

A possible extension of Riemannian geometry is the Riemann-Cartan
geometry in which, besides the metric tensor, there is an extra
geometrical field, the torsion tensor.  The Riemann
curvature tensor depends now
on both metric and torsion.  Moreover, the metric compatible
affine connection, between
tensors, contains not only a symmetric part as
in Riemannian geometry, but also an antisymmetric part, which is
precisely the torsion tensor. In a Riemann-Cartan geometry, geodesic
and autoparallel curves are different types of curves.  

A natural
extension of general relativity to another theory of gravitation is
the Einstein-Cartan theory where the underlying geometry is the
Riemann-Cartan geometry.  
A realization of an Einstein-Cartan gravity theory is such that
the field equations are still derived from the Einstein-Hilbert action
\citep{Sciama_Book,Kibble_1961}, representing one of the simplest
generalizations of general relativity. The link between geometry
and matter is now
provided by the Einstein equation, that equates the Einstein tensor to the
canonical
matter stress-energy tensor, plus an equation relating a tensor field built
out of the geometrical torsion to some physical observables associated,
for instance,
to the density of the intrinsic angular momentum of matter, or spin.
One of the interests in the Einstein-Cartan theory, 
within a geometric theory of gravitation, is that at extremely high
densities of matter, even at densities still much less than the Planck
density regime where quantum gravity rules, quantum effects on the
matter may be considerable, hence the ability to include quantum corrections
in a macroscopic limit, through the relation between torsion
and intrinsic spin,
might set the Einstein-Cartan theory to be a better
classical limit of a quantum theory
of gravitation than a theory without torsion like general relativity.
Nonetheless, as
we will show in this article, even in the low density regime, the inclusion of
torsion might lead to marked contrast between the predictions
of the two theories,
which may be used to falsify the hypothesis.

The framework of the Einstein-Cartan theory has led to many important
results, showing how the extra geometrical structure, specifically,
the torsion tensor, influences the
behavior of the matter fluids that permeate the spacetime and,
consequently, the geometry of the manifold.  Several works have worked
out the properties of spacetimes with torsion and the consequences of
the Einstein-Cartan theory, in particular in black hole physics and in
cosmology.  We mention some of them.  The possibility of measuring
torsion was raised in \cite{Hehl_1971}.
Some works showed that the
inclusion of torsion could act as a repulsive force, counteracting the
gravitational collapse and possibly prevent the formation of
singularities both in black holes and cosmology
\citep{Trautman_1973,Stewart_Hajicek_1973}.  There were applications
in cosmology \cite{Kopczynski_1973,Tafel_1973} as well as in rotating
neutron stars \cite{Kerlick_1973}, see also
\cite{Hehl_Heyde_Kerlick_1974,Hehl_Heyde_Kerlick_Nester_1976}.
The generation of solutions in Einstein-Cartan theory was provided in
\cite{Tsamparlis_1981}.
Some
interesting spinning fluids, in particular the Weyssenhoff fluid, were
introduced as generators of torsion in \cite{Ray_Smalley_1983}, as
providers of inflation in \cite{Gasperini_1986}, and as sources of
rotating cosmological models in \cite{Obukhov_Korotky_1987}.
The possibility that the Einstein-Cartan theory is a limit
of a quantum theory of the
gravitational field operating at
the usual
microscopic and macroscopic scales
has been hinted in \citep{Trautman_inBook_2006}.
The consequences and imprints on the curves followed by finite size
test bodies was discussed in \cite{Puetzfeld_Obukhov_2007}.  A review
of Einstein-Cartan theory is in \citep{Blagojevic_Hehl_book_2013}.  A
further discussion on compact objects was given in
\cite{Luz_Sante_2019}, and a study of the cosmological signatures of
torsion and cosmic acceleration appears in
\cite{Bolejko_Cinus_Roukema_2020,Iosifidis_2021}.

\subsection{The 1+3 spacetime decomposition}

Due to the action of the torsion tensor in physical frames of
reference, in gravitational theories with torsion, it is advantageous
to work in a formalism that is manifestly covariant and such that the
quantities that characterize the spacetime and the matter are directly
associated with physical observables.  A formalism with such
characteristics is the covariant spacetime decomposition approach
which is designed to take directly into account the symmetries and
preferred directions in a manifold, and emerges as a powerful tool to
analyze the geometry and dynamics of tensor fields on a spacetime.
A benefit one takes from the formalism
comes from the fact that it is, by construction, independent of 
coordinate systems. Moreover, the natural
splitting of the manifold can greatly simplify the problem of finding
solutions when the spacetime admits the existence of preferred directions,
such as Killing vector fields.

A particular covariant spacetime decomposition is the 
1+3 formalism that has been developed and used in
many instances in 
general relativity
\citep{Ehlers_1961,
Stewart_Ellis_1968,Ellis_1971,Ellis_vanElst_1999,
Ellis_Maartens_MacCallum_Book}.
In this formalism, it is assumed the existence of a congruence of
smooth curves, so that any tensor quantity on the spacetime is, at
each point, separated in its component along the direction of the
tangent vector field to the congruence and
in its components along the surfaces orthogonal
to the curves of the congruence. This property
of the formalism makes it especially useful from a physical point
of view, since in many instances
one is interested in studying the evolution
of certain quantities in time. Thus, assuming the existence of a timelike
congruence, the 1+3 formalism naturally decomposes the structure equations
that describe the geometry of the spacetime and tensor fields in
the manifold along time and spatial directions.
In a geometric theory of gravity, the geometry
of the spacetime is related with the matter fields that permeates
it. Since the evolution and constraint equations found from the 1+3
formalism are completely covariant, they provide a clearer interpretation
of the relations between the kinematics of the congruence and the
properties of the matter fields. 

This formalism of 1+3 decomposition of the spacetime manifold has been
extensively employed to explore the properties of solutions of
theories of gravitation, namely, gravitational waves, cosmological
solutions, compact objects, black holes, singularities, and particle
and light rays propagation.  For instance, the formalism has been used
in general relativity to study cosmological perturbations and their
consequent gravitational waves generation 
\cite{Dunsby_Bassett_Ellis_1997}, to analyse singularities and
singularity theorems  \cite{Senovilla_1998}, to develop an effective
fluid dynamics formalism  \cite{Brechet_Hobson_Lasenby_(2007)}, to
find new properties of perturbed Schwarzschild black holes
\cite{Clarkson_Barrett_2003}, to further investigate the
Tolman-Oppenheimer-Volkoff equation \cite{Carloni_Vernieri_2008}, to
discuss cosmological perfect fluid perturbations 
\cite{Tornkvist_Bradley_2019}, and to analyze objects composed of two
fluids  \cite{Naidu_Carloni_Dunsby_2021}.
The formalism has also been
applied in $f\left(R\right)$ modified theories of
gravitation to study gravitational lensing  
\cite{Dunsby_Goswami_2011},
to introduce
black holes with emphasis in the Weyl terms  \cite{Pratten_2015},
to describe cosmological density perturbations \cite{Carloni_Troisi_2008},
and to search for
gravitational wave solutions \cite{Ananda_Carloni_2008}.
The formalism has further been
applied in theories with torsion
to explore the 
Raychaudhuri
equation \cite{Luz_Vitagliano_2017},
to treat
spacetime thermodynamics \cite{Dey_Liberati_Pranzetti_2017},
to examine
singularity theorems \cite{Luz_Mena_2020},
and to establish focusing condition theorems
\cite{Vinckers_Torralba_2020}.

\subsection{Aim of the work}

Since Hubble showed the velocity-distance relation for distant
galaxies that we know the universe is expanding, and the observations
of the emission spectra of type Ia supernovae have lead to the
conclusion that our universe is expanding at an accelerated rate.
Moreover, the high-precision data from the Hubble Space Telescope,
WMAP, Planck collaboration, Sloan Digital Sky Survey, and JWST, keep
confirming that, at very large scales, the present universe is very
well described by the Friedman-Lema\^itre-Robertson-Walker (FLRW) model
with the matter source having as a main
component an unknown dark energy fluid.
The other component, with its existence also being infered
from the data from rotational curves of
galaxies and from the velocities of individual members
of galaxies in clusters,
could be in unknown forms of dark particles 
or showing 
that the predictions of the
general theory of relativity, even at a classical level, are
incompatible with the observations without the inclusion of extra
fundamental 
fields in alternative gravity theories.
On the other hand, high-accuracy astrometric data of GAIA,
and the regular detection of gravitational waves impose very strong
constraints on these possible alternative gravity theories,
notably on theories of the Einstein-Cartan type.

In this context,
there has been a growing interest in studying the effects of torsion
in the dynamics of the universe by considering types of torsion that
have the intrinsic spin of matter as its source.
Here, we are interested in studying the effects of the intrinsic spin
of matter, within the canonical Einstein-Cartan theory, in the
properties and evolution of the universe at very large scales.  We
will show that there are various aspects that have been overlooked, in
particular we will show that in the Einstein-Cartan theory it is
pivotal to understand the effects of the torsion field in the Weyl
tensor, which to our knowledge have never been considered. As it will
be shown, the coupling between the Weyl
tensor and the torsion tensor field may lead to dramatic disparities
between the predictions of the general theory of relativity and the
Einstein-Cartan theory.
Moreover, although the 1+3 formalism was initially devised to study
solutions of general relativity, the paradigm of covariant spacetime
decomposition is applicable to a much wider class of relativistic
theories of gravitation, including theories of the Einstein-Cartan
type.  We will also present here the most general extension of the 1+3
formalism for spacetimes with a metric compatible affine connection
valid for any metric affine gravity theory. These equations will then
be used in the particular case of the Einstein-Cartan theory and will
be used to study the effects of intrinsic spin in spacetimes permeated
by an homogeneous and isotropic matter fluid.

\subsection{Organization of the work}

This work is organized as follows.
In Sec.~\ref{Section:Conventions}, we introduce several quantities,
namely, the metric, the spacetime connection with torsion, the
curvature, the projection formalism with the decomposition of the
torsion and Weyl tensors and the structure equations, giving as well
the basic definitions and setting the adopted conventions.
In Sec.~\ref{Section:EC_stressenergytensor},  the stress-energy tensor
and the structure equations for the matter fields are given.
In Sec.~\ref{Section:EC_structure_equations},  the
field equations of
the Einstein-Cartan theory for a Weyssenhoff like torsion are found
and compared with the results in the literature.
In Sec.~\ref{Section:Isotropic_universe},  the isotropic universe and
the geometry of the 3-spaces are set as a basis for a relativistic
cosmology in
the new set of equations for the 
Einstein-Cartan theory
for  a universe
permeated by an isotropic and homogeneous matter fluid with nonvanishing
intrinsic spin, and where two theorems and a
proposition are proved.
In Sec.~\ref{Subsec:Wave-equation},  gravitational waves in
relativitstic cosmology in Einstein-Cartan theory are studied.
In Sec.~\ref{Subsec:Tidal-effects-and-dynamics}, we analyze
the tidal effects and the
dynamics of the cosmic fluid in relativistic cosmology in
Einstein-Cartan theory.
In Sec.~\ref{Section:Conclusion},  we further discuss the main results
and conclude.
In the Appendix~\ref{laplacebeltramiharmonics},  we display the
properties of the Laplace-Beltrami harmonics which are used in the
main text.
Throughout the paper we will work in the geometrized unit system where
the constant of gravitation and the speed of light are set to one,
and consider the metric signature $(-+++)$.

\section{Geometry of Lorentzian manifolds with torsion
and the structure equations
for the geometric fields
\label{Section:Conventions}}

\subsection{Metric, connection with torsion, curvature, and projection
formalism}
\subsubsection{Metric, connection with torsion, curvature}

We start by introducing the basic definitions and setting the
conventions that will be used throughout this article.

Let $\left(\mathcal{M},g,S\right)$ be a 4-dimensional
Lorentzian manifold endowed with a metric compatible affine connection.
The metric tensor $g$ is assumed to be symmetric, i.e.,
\begin{equation}
g_{\alpha\beta}=
g_{\left(\alpha\beta\right)}\,,
\label{Conventions_eq:metric}
\end{equation}
with
$g_{\left(\alpha\beta\right)}
\equiv
\frac{1}{2}\left( g_{\alpha\beta}+
g_{\beta\alpha}\right)$,
and the tensor $S$ represents the torsion tensor, defined as the
antisymmetric part of the connection in the lower indices,
and is such that
\begin{equation}
{S_{\alpha\beta}}^\gamma=
{S_{\left[\alpha\beta\right]}}^\gamma\,,
\label{Conventions_eq:torsion}
\end{equation}
with
${S_{\left[\alpha\beta\right]}}^\gamma
\equiv\frac12\left(
{S_{\alpha\beta}}^\gamma-
{S_{\beta\alpha}}^\gamma\right)$.
In 
$\left(\mathcal{M},g,S\right)$
the covariant derivative $\nabla$
is defined 
through an affine connection $C_{\alpha \beta}^\gamma$,
such that on a
$\left( k,m\right)-$tensor $Y$ with components in a local coordinate
system $Y^{\mu_{1}\ldots\mu_{k}}{}_{\nu_{1}\ldots\nu_{m}}
$, is formally given by
$\nabla_{\alpha}Y^{\mu_{1}\ldots\mu_{k}}{}_{\nu_{1}\ldots\nu_{m}}=
\partial_{\alpha}Y^{\mu_{1}\ldots\mu_{k}}{}_{\nu_{1}\ldots\nu_{m}}
+
\sum_{i=1}^{k}
C_{\alpha\rho}^{\mu_{i}}{}\,Y^{\mu_{1}\ldots
\rho\ldots\mu_{k}}{}_{\nu_{1}\ldots\nu_{m}}
-
\sum_{i=1}^{m}
C_{\alpha\nu_{i}}^\rho{}\,Y^{\mu_{1}
\ldots\mu_{k}}{}_{\nu_{1}\ldots\rho\ldots\nu_{m}}
$,
where $\partial_{\alpha}$ represent partial derivatives.
In order that 
the affine connection $C_{\alpha \beta}^\gamma$ be metric $g$ compatible,
one has that $\nabla_{\alpha}g_{\beta\gamma}=
\partial_\alpha g_{\beta\gamma}-
C_{\alpha\beta}^\sigma g_{\sigma\gamma}
-
C_{\alpha\gamma}^\sigma g_{\beta\sigma}
$ has to be identically zero, i.e.,
\begin{equation}
\nabla_{\alpha}g_{\beta\gamma}=0\,.
\label{Conventions_eq:Covariant_definition}
\end{equation}
A metric compatible connection $C_{\alpha \beta}^\gamma$
can always be split into two parts, namely,
\begin{equation}
C_{\alpha\beta}^{\gamma}=\Gamma_{\alpha\beta}^{\gamma}+
K_{\alpha\beta}{}^{\gamma}\,,
\label{Conventions_eq:Connection}
\end{equation}
with 
\begin{equation}
\Gamma_{\alpha\beta}^{\gamma}=\frac{1}{2}g^{\gamma\sigma}
\left(\partial_{\alpha}g_{\sigma\beta}+
\partial_{\beta}g_{\alpha\sigma}-\partial_{\sigma}g_{\alpha\beta}
\right)\,,
\label{Conventions_eq:Christoffel_symbols}
\end{equation}
being the  usual
metric connection that appears in a Riemannian manifold,
also referred as the Christoffel symbols,
and 
\begin{equation}
K_{\alpha\beta}{}^{\gamma}= S_{\alpha\beta}{}^{\gamma}+
S^{\gamma}{}_{\alpha\beta}-
S_{\beta}{}^{\gamma}{}_{\alpha}\,,
\label{Conventions_eq:Contorsion}
\end{equation}
being the contorsion tensor
which is a combination of torsion terms.
From
Eqs.~\eqref{Conventions_eq:Connection}--\eqref{Conventions_eq:Contorsion}
one finds
\begin{equation}
\Gamma_{\alpha\beta}^{\gamma}=
\Gamma_{\left(\alpha\beta\right)}^{\gamma}=
C_{\left(\alpha\beta\right)}^{\gamma}\,,
\label{Conventions_eq:Gamma}
\end{equation}
and 
\begin{equation}
S_{\alpha\beta}{}^{\gamma}
=K_{\left[\alpha\beta\right]}{}^{\gamma}=
C_{\left[\alpha\beta\right]}^{\gamma}\,.
\label{Conventions_eq:Torsion_tensor}
\end{equation}
As well, 
from the antisymmetry of the torsion tensor in the first two
indices, one can verify the following identity for
the contorsion tensor,
\begin{equation}
K_{\alpha\beta\gamma}=K_{\alpha\left[\beta\gamma\right]}
\,,
\label{Conventions_eq:contorsion_identities}
\end{equation}
i.e., 
$
K_{\alpha\left(\beta\gamma\right)}=0$.

The definition of the Riemann curvature tensor associated
with the connection $C_{\alpha\beta}^{\gamma}$ is 
\begin{equation}
R_{\alpha\beta\gamma}\,^{\rho}=
\partial_{\beta}C_{\alpha\gamma}^{\rho}-
\partial_{\alpha}C_{\beta\gamma}^{\rho}+
C_{\beta\sigma}^{\rho}C_{\alpha\gamma}^{\sigma}-
C_{\alpha\sigma}^{\rho}C_{\beta\gamma}^{\sigma}\,.
\label{Conventions_eq:Riemann_tensor_connnection}
\end{equation}
This definition
leads to the following relation between
the commutator of two covariant derivatives
of a tensor 
and the Riemann curvature
tensor, Eq.~\eqref{Conventions_eq:Riemann_tensor_connnection},
\begin{align}
\left(\nabla_{\alpha}
\nabla_{\beta}-\nabla_{\beta}\nabla_{\alpha}+
2S_{\alpha\beta}{}^{\gamma}\nabla_{\gamma}\right)
Y^{\mu_{1}\ldots\mu_{k}}{}_{\nu_{1}\ldots\nu_{m}}=&
\sum_{i=1}^{m}R_{\alpha\beta\nu_{i}}{}^{\rho}\,Y^{\mu_{1}
\ldots\mu_{k}}{}_{\nu_{1}\ldots\rho\ldots\nu_{m}}-
\sum_{i=1}^{k}R_{\alpha\beta\rho}{}^{\mu_{i}}\,Y^{\mu_{1}\ldots
\rho\ldots\mu_{k}}{}_{\nu_{1}\ldots\nu_{m}}
\,,
\label{Conventions_eq:Riemann_tensor_definition}
\end{align}
where $Y$ is an arbitrary $\left(k,m\right)$-tensor field.
The Riemann curvature
tensor, Eq.~\eqref{Conventions_eq:Riemann_tensor_connnection}, possesses
the following symmetries in its indices,
\begin{align}
R_{\alpha\beta\gamma\delta}= R_{[\alpha\beta]\gamma\delta}\,,
\label{Conventions_eq:Riemann_tensor_properties1}
\end{align}
i.e., $R_{(\alpha\beta)\gamma\delta}=0$, and 
\begin{align}
R_{\alpha\beta\gamma\delta}= R_{\alpha\beta[\gamma\delta]}\
\,,
\label{Conventions_eq:Riemann_tensor_properties2}
\end{align}
i.e., $R_{\alpha\beta(\gamma\delta)}=0$. The symmetries
of the Riemann curvature
tensor given in Eqs.~\eqref{Conventions_eq:Riemann_tensor_properties1}
and \eqref{Conventions_eq:Riemann_tensor_properties2}
are the same as in pure Riemannian geometry.
The other index symmetry in pure Riemannian geometry, namely,
$R_{\left[\alpha\beta\gamma\right]}{}^{\delta}=0$, is,
for a geometry with torsion, modified
into an identity related to the
covariant derivative of the torsion,
\begin{align}
2\nabla_{\left[\alpha\right.}S_{\left.\beta\gamma\right]}{}^{\delta}
-
4S_{[\alpha\beta}{}^{\rho}
S_{\gamma]\rho}{}^{\delta}+
R_{\left[\alpha\beta\gamma\right]}{}^{\delta}  =0
\,,\label{Conventions_eq:first_Bianchi_identity}
\end{align}
which can be envisaged as a Bianchi identity for the torsion $S$,
and is in this context called the first Bianchi identity.
We note that the antisymetrization in the second term
of Eq.~\eqref{Conventions_eq:first_Bianchi_identity}
only refers to nondummy indices,
in this case to $\alpha,\beta,\gamma$, with
the dummy index
$\rho$ not being affected by the process, and this is
a convention that we will follow. 
The Riemannian Bianchi identity, namely 
$\nabla_{\left[\alpha\right.}R_{\left.\beta\gamma\right]\delta}{}^{\rho}=0$,
when torsion is present is modified into 
\begin{align}
\nabla_{\left[\alpha\right.}R_{\left.\beta\gamma\right]\delta}{}^{\rho} 
-2S_{[\alpha\beta}{}^{\sigma}R_{\gamma]
\sigma\delta}{}^{\rho}=0\,,
\label{Conventions_eq:second_Bianchi_identity}
\end{align}
and is in this context called the second Bianchi identity.
From the index symmetry identities,
Eqs.~\eqref{Conventions_eq:Riemann_tensor_properties1}
and
\eqref{Conventions_eq:Riemann_tensor_properties2},
and the
first Bianchi
identity, Eq.~(\ref{Conventions_eq:first_Bianchi_identity}),
we find that the usual symmetry of exchanging the first and second
pair of indices of the Riemann tensor is modified in the presence
of torsion to
\begin{equation}
2R_{\gamma\delta\alpha\beta}=2R_{\alpha\beta\gamma\delta}+
3A_{\alpha\gamma\beta\delta}+3A_{\delta\alpha\beta\gamma}+
3A_{\gamma\delta\alpha\beta}+3A_{\beta\delta\gamma\alpha}\,,
\label{Conventions_eq:Riemann_exchange_pair_indices}
\end{equation}
where we have
written $A_{\alpha\beta\gamma\delta}\equiv
-2\nabla_{[\alpha}S_{\beta\gamma]\delta}
+
4S_{[\alpha\beta}{}^{\rho}S_{\gamma]\rho\delta}$
to simplify the visualization of the equation.
We remark that the results presented so far are completely general,
in particular, they are
valid for spacetimes of any dimension.

We will now consider the case
of an orientable Lorentzian manifold $\left(\mathcal{M},g,S\right)$ of
dimension 4. In this case,
a useful quantity to define is the Levi-Civita volume
form, also referred as covariant Levi-Civita tensor or
Levi-Civita 4-form.
Introducing the Levi-Civita symbol, $\eta_{\alpha\beta\gamma\delta}$,
as the totally skew tensor density whose components in any orientation
preserving local coordinate system verify $\eta_{1234}=+1$, the Levi-Civita
volume form is defined as
\begin{equation}
\varepsilon_{\alpha\beta\gamma\delta}\equiv\sqrt{\left|\det g\right|}\,
\eta_{\alpha\beta\gamma\delta}\,,
\label{Conventions_eq:Levi-Civita_tensor_lower}
\end{equation}
where $\left|\det g\right|$ represents the absolute value of the
determinant of the metric tensor. The Levi-Civita volume form verifies
some  useful relations, namely,
(i)~$\nabla_{\rho}\varepsilon_{\alpha\beta\gamma\delta}=0$,
(ii)~$\varepsilon^{\alpha\beta\gamma\delta} =
\frac{\text{sign}\left(\det g\right)}{\sqrt{\left|\det g\right|}}
\eta^{\alpha\beta\gamma\delta}$,
(iii)~$\varepsilon_{\alpha\beta\gamma\delta}\varepsilon^{\rho\sigma\mu\nu}
=-24\,g^{\rho}{}_{\left[\alpha\right.}
g^{\sigma}{}_{\beta}g^{\mu}{}_{\gamma}g_{\left.\delta\right]}{}^{\nu}$,
(iv)~$\varepsilon_{\alpha\beta\gamma\delta}\varepsilon^{\alpha\sigma\mu\nu}
= -6\,g^{\sigma}{}_{\left[\beta
\right.}g^{\mu}{}_{\gamma}g_{\left.\delta\right]}{}^{\nu}$,
(v)~$ \varepsilon_{\alpha\beta\gamma\delta}\varepsilon^{\alpha\beta\mu\nu}
= -4\,g^{\mu}{}_{\left[\gamma\right.}g_{\left.\delta\right]}{}^{\nu}
$, and (vi)~$\varepsilon_{\alpha\beta\gamma\delta}
\varepsilon^{\alpha\beta\gamma\delta} 
=-24$.
The first equality follows from the assumption that the
connection is metric compatible, the second from the properties of the
determinant of a matrix and in (iii) to (vi) only the lower indices
are to be antisymmetrized.

The Weyl tensor $C_{\alpha\beta\gamma\delta}$ is
defined as the trace-free part of the Riemann
curvature
tensor
$R_{\alpha\beta\gamma\delta}$.
In the case of a manifold of dimension 4, the components of the
Weyl curvature tensor, $C_{\alpha\beta\gamma\delta}$, can be
written as 
\begin{equation}
C_{\alpha\beta\gamma\delta}=
R_{\alpha\beta\gamma\delta}-
R_{\alpha\left[\gamma\right.}g_{\left.\delta\right]\beta}+
R_{\beta\left[\gamma\right.}g_{\left.\delta\right]\alpha}+
\frac{1}{3}R\,g_{\alpha\left[\gamma\right.}
g_{\left.\delta\right]\beta}\,,
\label{Conventions_eq:Weyl_tensor_definition}
\end{equation}
where 
$R_{\alpha\beta}\equiv R_{\alpha\mu\beta}{}^{\mu}$ is
the Ricci tensor, and $R\equiv R_{\mu}{}^{\mu}$
is the Ricci scalar.
The Weyl tensor
inherits, from Eq.~\eqref{Conventions_eq:Weyl_tensor_definition},
the following symmetries in its indices,
\begin{align}
C_{\alpha\beta\gamma\delta}= C_{[\alpha\beta]\gamma\delta}\,,
\label{Conventions_eq:Weyl_tensor_properties}
\end{align}
i.e., $C_{(\alpha\beta)\gamma\delta}=0$, and 
\begin{align}
C_{\alpha\beta\gamma\delta}= C_{\alpha\beta[\gamma\delta]}\
\,,
\label{Conventions_eq:Weyl_tensor_properties2}
\end{align}
i.e., $C_{\alpha\beta(\gamma\delta)}=0$. 
In addition, one finds
\begin{align}
C_{\left[\alpha\beta\gamma\right]\delta}  =
R_{\left[\alpha\beta\gamma\right]\delta}+
R_{\left[\alpha\beta\right.}g_{\left.\gamma\right]\delta}
\,.
\label{Conventions_eq:Weyl_tensor_properties3}
\end{align}
In the presence of torsion, the relation between the derivative
of the Weyl tensor and the Riemann tensor is~\citep{Luz_Sante_2019}
\begin{align}
\nabla_{\alpha}C^{\gamma\delta\beta\alpha}= 
\frac{1}{2}\,\varepsilon^{\mu\nu\lambda\beta}
S_{[\mu\nu}{}^{\sigma}R_{\lambda]\sigma\eta\rho}
\varepsilon^{\eta\rho\gamma\delta}+
\frac{3}{2}\left(g^{\beta\delta}S^{[\gamma\mu}{}_{\sigma}R^{\nu]
\sigma}{}_{\mu\nu}-g^{\beta\gamma}S^{[\delta\mu}{}_{\sigma}R^{\nu]
\sigma}{}_{\mu\nu}\right)
+
 \nabla^{[\delta}R^{\gamma]\beta}-
 \frac{1}{6}\,g^{\beta[\gamma}\nabla^{\delta]}R\,,
\label{Conventions_eq:Div_Weyl_trumper}
\end{align}
and
the dummy index $\sigma$ is not involved in the antisymmetrization
process.

\subsubsection{Projector operator, projected covariant Levi-Civita
tensor, and projected covariant derivatives}

Consider a Lorentzian manifold of dimension 4, $\left(\mathcal{M},g,S\right)$,
admitting, in some open neighborhood, the existence of a congruence
of timelike curves with tangent vector field $u$. Without loss of
generality, we can foliate the manifold in 3-surfaces, $\mathcal{V}$,
orthogonal, at each point, to the curves of the congruence, such that
all tensor quantities are defined by their behavior along the direction
of $u$ and in $\mathcal{V}$. This procedure is usually called 1+3
spacetime decomposition. Such decomposition of the spacetime manifold
relies on the existence of a projector to the hypersurface $\mathcal{V}$.
Assuming each curve of the congruence to be affinely parameterized,
so that $u_{\alpha}u^{\alpha}=-1$, the projector onto $\mathcal{V}$,
at each point can be defined as
\begin{equation}
h_{\alpha\beta}\equiv g_{\alpha\beta}+u_{\alpha}u_{\beta}\,,
\label{Projector_eq:orthogonal_projector_timelike}
\end{equation}
verifying
$h_{\alpha\beta}u^{\alpha}  =0$,
$h_{\alpha\beta}  =h_{\beta\alpha}$, 
$h_{\alpha}{}^{\gamma}h_{\gamma\beta}  =h_{\alpha\beta}$,
and
$h_{\gamma}{}^{\gamma}  =3$.

Another useful operator is the projected covariant Levi-Civita tensor
\begin{equation}
\varepsilon_{\alpha\beta\gamma}=
\varepsilon_{\alpha\beta\gamma\sigma}u^{\sigma}\,,
\label{Projector_eq:Projected_1+3_Levi-Civita}
\end{equation}
derived from the Levi-Civita volume form, defined in
Eq.~(\ref{Conventions_eq:Levi-Civita_tensor_lower}),
with the following properties
$\varepsilon_{\alpha\beta\gamma}
=\varepsilon_{\left[\alpha\beta\gamma\right]}$,
$\;\varepsilon_{\alpha\beta\gamma}u^{\gamma} =0$,
$\;\varepsilon_{\alpha\beta\gamma}\varepsilon^{\mu\nu\sigma}
=6h^{\mu}{}_{\left[\alpha\right.}h^{\nu}{}_{\beta}h_{\left.
\gamma\right]}{}^{\sigma}$,
$\;\varepsilon_{\alpha\beta\gamma}\varepsilon^{\mu\nu\gamma}
=2h^{\mu}{}_{\left[\alpha\right.}h_{\left.\beta\right]}{}^{\nu} $,
$\;\varepsilon_{\mu\nu\alpha}\varepsilon^{\mu\nu\beta}
=2h_{\alpha}{}^{\beta}$.
Moreover, using Eq.~(\ref{Projector_eq:orthogonal_projector_timelike})
and the properties of the Levi-Civita volume form, we find the useful
identities
$h_{\mu}{}^{\alpha}h_{\nu}{}^{\beta}
h_{\rho}{}^{\gamma}h_{\lambda}{}^{\sigma}
\varepsilon_{\alpha\beta\gamma\sigma}
=0$
and 
$\varepsilon_{\alpha\beta\gamma\sigma}  =
h_{\alpha}{}^{\mu}\varepsilon_{\mu\beta\gamma\sigma}+
u_{\alpha}\varepsilon_{\beta\gamma\sigma}$.

In order to keep the equations as compact as possible, we
introduce
the following notation
for projected covariant derivatives. Given a tensor field
$Y_{\alpha...\beta}{}^{\gamma...\sigma}$
we define 
\begin{equation}
D_{\mu}Y_{\alpha...\beta}{}^{\gamma...
\sigma}\equiv h_{\mu}{}^{\nu}h_{\alpha}{}^{\rho}...
h_{\beta}{}^{\delta}h_{\lambda}{}^{\gamma}...
h_{\varphi}{}^{\sigma}\nabla_{\nu}Y_{\rho...
\delta}{}^{\lambda...\varphi}\,,
\label{Ddefinition}
\end{equation}
as the fully orthogonally projected covariant derivative
on $\mathcal{V}$.
On the other hand, a dot represents the covariant derivative
along the integral curves of $u$, i.e., 
\begin{equation}
\dot{Y}_{\alpha...\beta}{}^{\gamma...
\sigma}=u^{\mu}\nabla_{\mu}Y_{\alpha...\beta}{}^{\gamma...\sigma}\,.
\label{dotdefinition}
\end{equation}

\subsection{Decomposition of the torsion tensor
and  Weyl tensor}

We now write
the 1+3 decomposition of the torsion tensor $S$,
Eq.~(\ref{Conventions_eq:Torsion_tensor}) and 
the Weyl tensor $C$, Eq.~(\ref{Conventions_eq:Weyl_tensor_definition}),
in terms of their components along the direction of the tangent vector
field $u$ and on $\mathcal{V}$ with the help
of the projector operator $h_{\alpha\beta}$ given in
Eq.~\eqref{Projector_eq:orthogonal_projector_timelike}.

For the torsion tensor, Eq.~(\ref{Conventions_eq:Torsion_tensor}),
the decomposition is~\citep{Luz_Sante_2019}
\begin{equation}
S_{\alpha\beta\gamma}=\varepsilon_{\alpha\beta}{}^{\mu}
\bar{S}_{\mu\gamma}-
u_{[\alpha} W_{\beta]\gamma}+
S_{\alpha\beta}u_{\gamma}+u_{[\alpha}X_{\beta]}u_{\gamma}\,,
\label{SWT_eq:timelike_torsion_decomposition}
\end{equation}
with
\begin{align}
&\bar{S}_{\alpha\beta} =
\frac{1}{2}\varepsilon_{\alpha\mu\nu}
h^{\sigma}{}_{\beta}S^{\mu\nu}{}_{\sigma}\,,\quad\quad
W_{\alpha\beta}  =2u^{\mu}h^{\nu}{}_{\alpha}
h^{\sigma}{}_{\beta}S_{\mu\nu\sigma}\,,\quad\nonumber\\
&S_{\alpha\beta} =-h_{\alpha}{}^{\mu}
h^{\nu}{}_{\beta}u^{\sigma}S_{\mu\nu\sigma}\,,\quad
X_{\alpha} =2u^{\mu}h^{\nu}{}_{\alpha}
u^{\sigma}S_{\mu\nu\sigma}\,.
\label{SWT_eq:timelike_torsion_decomposition_components_def}
\end{align}

For the Weyl tensor, Eq.~(\ref{Conventions_eq:Weyl_tensor_definition}),
the 1+3 decomposition is
\begin{equation}
C_{\alpha\beta\gamma\delta}=-\varepsilon_{\alpha\beta\mu}
\varepsilon_{\gamma\delta\nu}E^{\nu\mu}-
2u_{\alpha}E_{\beta\left[\gamma\right.}
u_{\left.\delta\right]}+2u_{\beta}E_{\alpha\left[\gamma\right.}
u_{\left.\delta\right]}-2\varepsilon_{\alpha\beta\mu}
H^{\mu}{}_{\left[\gamma\right.}u_{\left.\delta\right]}-
2\varepsilon_{\mu\gamma\delta}\bar{H}^{\mu}{}_{\left[\alpha\right.}
u_{\left.\beta\right]}\,,
\label{SWT_eq:Weyl_tensor_1+3_decomposition}
\end{equation}
where 
\begin{align}
E_{\alpha\beta}  =
C_{\alpha\mu\beta\nu}u^{\mu}u^{\nu}\,,\quad
H_{\alpha\beta}  =\frac{1}{2}
\varepsilon_{\alpha}{}^{\mu\nu}C_{\mu\nu\beta\delta}u^{\delta}\,,\quad
\bar{H}_{\alpha\beta} & =
\frac{1}{2}\varepsilon_{\alpha}{}^{\mu\nu}
C_{\beta\delta\mu\nu}u^{\delta}\,,
\label{SWT_eq:Weyl_tensor_magnetic_bar}
\end{align}
are defined as the electric part and the
magnetic part of the Weyl tensor,
respectively. In the Riemann-Cartan geometry
there
are two different tensor quantities associated to the magnetic part
of the Weyl tensor, specifically,
$H_{\alpha\beta}$ and $\bar{H}_{\alpha\beta}$,
such that, in general the presence of torsion lifts a degeneracy in
the magnetic part of the Weyl tensor.
From the results in
Eqs.~\eqref{Conventions_eq:Weyl_tensor_properties}--\eqref{Conventions_eq:Weyl_tensor_properties3},
we see that in the presence of torsion the tensor $E_{\alpha\beta}$
 has the following
properties,
$E_{\alpha\beta} =h_{\alpha}{}^{\mu}h_{\beta}{}^{\nu}E_{\mu\nu}$
and $E^{\alpha}{}_{\alpha}=0$,
the tensor
$H_{\alpha\beta}$ has the following
properties
$H_{\alpha\beta}  =h_{\alpha}{}^{\mu}
h_{\beta}{}^{\nu}H_{\mu\nu}$ and 
$H_{\alpha\beta}  =H_{\left(\alpha\beta\right)}$, 
and the tensor
$\bar{H}_{\alpha\beta}$ 
has the following
properties
$\bar{H}_{\alpha\beta}  =h_{\alpha}{}^{\mu}
h_{\beta}{}^{\nu}\bar{H}_{\mu\nu}$
and
$\bar{H}_{\alpha\beta}  =\bar{H}_{\left(\alpha\beta\right)}$.
Therefore, $E_{\alpha\beta}$ may not be a symmetric tensor and
$H_{\alpha\beta}$ and $\bar{H}_{\alpha\beta}$ do not have to be
trace-free, as in the case of Riemannian geometry. On the other hand,
due to the properties of the
Levi-Civita volume form and the fact that the Weyl tensor is, by
definition, trace free, even in the presence of torsion, the magnetic
parts, $H_{\alpha\beta}$ and $\bar{H}_{\alpha\beta}$, are symmetric
under the exchange of indices.

Many of the results
introduced and to be
introduced are valid or easily extended for manifolds of dimension
$d\geq2$. However, quantities and identities that rely on the
covariant Levi-Civita tensor,
Eq.~(\ref{Conventions_eq:Levi-Civita_tensor_lower}), notably the 1+3
decomposition of the torsion tensor,
Eq.~\eqref{SWT_eq:timelike_torsion_decomposition},
and of the 
Weyl tensor Eq.~\eqref{SWT_eq:Weyl_tensor_1+3_decomposition},
are dimension dependent, hence, the
general set of structure equations that we will
find will depend on the dimension of the
manifold.

\subsection{The separation vector between infinitesimally close curves
of a congruence\label{Subsec:The-separation-vector}}

Having introduced the definitions and properties of the basic
geometric quantities and their decompositions, we will now
consider the notion of separation
vector between infinitesimally close curves of a congruence and relate
its evolution with the kinematical quantities that characterize the
congruence.  For further details see
\citep{Luz_Vitagliano_2017}.

Consider a congruence of curves in some open neighborhood of
$\left(\mathcal{M},g,S\right)$, with tangent vector field
$u$. Given two points $p$ and $q$ in a small enough
neighborhood, such that $p$ is crossed by a curve of the congruence
and $q$ is crossed by a distinct curve of the congruence, the vector
field $n\equiv q-p$ gives a meaningful notion of the separation
between the curves of the congruence. Picking a curve $c$ of the
congruence as the fiducial curve, it is possible to derive an equation
for the change of the separation vector $n$ along the curve $c$.
Indeed, one finds
\begin{equation}
u^{\beta}\nabla_{\beta}n^{\alpha}=B_{\beta}{}^{\alpha}n^{\beta}\,,
\label{Separation_eq:derivative_n_relation}
\end{equation}
where 
\begin{equation}
B_{\beta}{}^{\alpha}=\nabla_{\beta}u^{\alpha}+
2S_{\gamma\beta}{}^{\alpha}u^{\gamma}\,.
\label{Separation_eq:B_tensor_general}
\end{equation}
The tensor $B$ gives 
the evolution of the separation
vector $n$ between two
infinitesimally close curves
along the fiducial curve.
We note that Eq.~(\ref{Separation_eq:derivative_n_relation})
is valid for the case
of $u$ being timelike, spacelike or null, with the fiducial curve
being a geodesic or not, although
we will be interested in the case of a timelike curve.
The first term in the right-hand side of
Eq.~(\ref{Separation_eq:B_tensor_general})
is the usual term present in pure Riemannian geometry,
while the second term in the right-hand side of
Eq.~(\ref{Separation_eq:B_tensor_general})
represents an explicit contribution of the torsion tensor to the evolution
of a congruence of curves.

We can now study some geometrical
implications of Eqs.~(\ref{Separation_eq:derivative_n_relation})
and (\ref{Separation_eq:B_tensor_general}).
Taking
the derivative along
$c$ of the quantity $n_{\alpha}u^{\alpha}$ reads
\begin{equation}
u^{\mu}\nabla_{\mu}\left(n_{\alpha}u^{\alpha}\right)=
n_{\beta}a^{\beta}+2S_{\sigma\gamma\alpha}
u^{\sigma}u^{\alpha}n^{\gamma}\,,
\label{Separation_eq:Dervivative_naUa}
\end{equation}
where the acceleration vector field $a$
in a local coordinate system has components given by 
\begin{equation}
a^{\alpha}=u^{\gamma}\nabla_{\gamma}u^{\alpha}\,.
\label{Separation_eq:acceleration_timelike}
\end{equation}
The expression given in Eq.~(\ref{Separation_eq:Dervivative_naUa})
can be seen as 
the failure of the separation vector $n$ and the tangent vector $u$
to stay orthogonal to each other. Indeed, if at a given point, $n$
and $u$ are orthogonal to each other, a
nonzero acceleration $a$
or a nonzero, general, torsion
$S$  will destroy the preservation
of such orthogonality along the curve.
Thus, this analysis of Eq.~(\ref{Separation_eq:Dervivative_naUa}) leads
to the conclusion that the tensor $B$, describing the behavior of
the separation vector might have, even in the
case of a zero acceleration $a$,
nonzero components tangential and orthogonal to the tangent
vector field associated with the fiducial curve $c$
when  torsion is present. Without loss
of generality, it is then possible to write $B_{\alpha\beta}$ in
terms of two components. One
component, $B_{\perp\alpha\beta}$,
is completely orthogonal to $u$,
and another component, $B_{\parallel\alpha\beta}$,
contains the remaining terms. Given a projector $h_{\alpha\beta}$
onto the surface orthogonal to the curve $c$ at a given point, we
can then write
\begin{equation}
B_{\alpha\beta}=B_{\perp\alpha\beta}+
B_{\parallel\alpha\beta}\,.
\label{Separation_eq:Bab_orth_tang_decomposition}
\end{equation}
Now, 
$B_{\perp\alpha\beta}$ is defined as 
$B_{\perp\alpha\beta}  \equiv  h_{\alpha}{}^{\gamma}
h_{\beta}{}^{\sigma}B_{\gamma\sigma}$.
Furthermore
we can define the kinematical quantities
of the congruence, namley, 
expansion $\theta$, shear $\sigma_{\alpha\beta}$,
and vorticity $\omega_{\alpha\beta}$,
of neighboring curves of the congruence
that only depend on the orthogonal part $B_{\perp}$ of the tensor
$B$, so that
we have the identity
$B_{\perp\alpha\beta}=\frac{h_{\alpha\beta}}{h_{\gamma}{}^{\gamma}}
\theta+\sigma_{\alpha\beta}+\omega_{\alpha\beta}$.
Since we are interested in 1+3 dimensions,
we have $h_{\gamma}{}^{\gamma}=3$, and so
\begin{equation}
B_{\perp\alpha\beta}=\frac13\,
h_{\alpha\beta}\,
\theta+\sigma_{\alpha\beta}+\omega_{\alpha\beta}\,.
\label{Separation_eq:B_orthl_esv_general}
\end{equation}
with
\begin{align}
\theta  =B_{\perp\gamma}{}^{\gamma}\,,\quad
\sigma_{\alpha\beta}  =B_{\perp\left(\alpha\beta\right)}-
\frac13\,h_{\alpha\beta}\,\theta\,,\quad
\omega_{\alpha\beta}  =B_{\perp\left[\alpha\beta\right]}\,,
\label{Separation_eq:kinematical quantities_general}
\end{align}
Then, of course, given a $B_{\perp\alpha\beta}$,
one uses Eq.~\eqref{Separation_eq:Bab_orth_tang_decomposition}
to determine $B_{\parallel\alpha\beta}$ as
$B_{\parallel\alpha\beta}=B_{\alpha\beta}-B_{\perp\alpha\beta}$.
The set of kinematical quantities $\theta$, $\sigma_{\alpha\beta}$,
and $\omega_{\alpha\beta}$, given in
Eq.~\eqref{Separation_eq:kinematical quantities_general}, characterize
a congruence in a Lorentzian manifold and represent one of the
building blocks of covariant spacetime decomposition approaches.
Note further that
the procedure that defines the projector operator
strictly depends on the specific family of curves considered, i.e.,
depends on the tangent vector field $u$.
Once
the projector is assigned, as, e.g., in
Eq.~\eqref{Projector_eq:orthogonal_projector_timelike}, one has that
Eq.~\eqref{Separation_eq:B_tensor_general}
together with Eq.~(\ref{Separation_eq:B_orthl_esv_general})
will give an actual
expression
for the  derivative of the tangent vector $u$
in terms of the kinematical quantities, the tangent vector itself,
its acceleration $a$, and the
torsion tensor $S$.

The results presented here are quite general and valid for curves of any
kind and easily extended to spacetimes of any dimension $d\geq 2$.
Nonetheless, in this work we will focus on developing the 1+3 formalism for
timelike congruences in a 4-dimensional oriented Lorentzian manifold
with torsion.

\subsection{Structure equations
for the geometric fields
\label{Section:1p3_structure_equations}}

The kinematical quantities of a congruence of curves
(\ref{Separation_eq:kinematical quantities_general}), the acceleration
vector field (\ref{Separation_eq:acceleration_timelike}) and the
tensors found from the decomposition of the torsion tensor,
Eqs.~\eqref{SWT_eq:timelike_torsion_decomposition} and
\eqref{SWT_eq:timelike_torsion_decomposition_components_def}, and of
the Weyl tensor, Eqs.~\eqref{SWT_eq:Weyl_tensor_1+3_decomposition} and
\eqref{SWT_eq:Weyl_tensor_magnetic_bar}, and the Ricci tensor
completely describe the
geometry of the manifold $\left(\mathcal{M},g,S\right)$ and the
properties of a congruence of curves that permeate it. We have then to
find a complete set of differential equations that describe the
evolution of these quantities along $u$ and on $\mathcal{V}$.

Now, projecting twice Eq.~\eqref{Separation_eq:B_tensor_general} with
the projector given in
Eq.~\eqref{Projector_eq:orthogonal_projector_timelike}, we find that
the covariant derivative of the tangent vector field $u$ is given by
$\nabla_{\alpha}u_{\beta}=
B_{\perp\alpha\beta}-W_{\alpha\beta}-u_{\alpha}a_{\beta}$, where we
have used
Eqs.~\eqref{SWT_eq:timelike_torsion_decomposition},
\eqref{Separation_eq:B_tensor_general}, and
\eqref{Separation_eq:acceleration_timelike}.  Then, using
Eq.~\eqref{Separation_eq:B_orthl_esv_general} we find
\begin{equation}
\nabla_{\alpha}u_{\beta}=\frac{1}{3}
h_{\alpha\beta}\theta+\sigma_{\alpha\beta}+
\omega_{\alpha\beta}-W_{\alpha\beta}-
u_{\alpha}a_{\beta}\,.
\label{SE_eq:Covariant_derivative_timelike_tangent}
\end{equation}
Applying the Ricci identity,
Eq.~(\ref{Conventions_eq:Riemann_tensor_definition}),
to Eq.~\eqref{SE_eq:Covariant_derivative_timelike_tangent},
we find the
propagation equations for the kinematical quantities
\begin{equation}
\begin{aligned}
\dot{\theta}-\dot{W}_{\alpha}{}^{\alpha}=
& -R_{\alpha\beta}u^{\alpha}u^{\beta}-\left(\frac{1}{3}
\theta^{2}+\sigma_{\alpha\beta}\sigma^{\alpha\beta}+
\omega_{\alpha\beta}\omega^{\beta\alpha}\right)\\
 & +D_{\alpha}a^{\alpha}+a_{\alpha}a^{\alpha}+W^{\beta\alpha}
 \left[\frac{1}{3}h_{\alpha\beta}\theta+
 \sigma_{\alpha\beta}+\omega_{\alpha\beta}\right]+
 X_{\alpha}a^{\alpha}\,,
\end{aligned}
\label{SE_eq::Raychaudhuri_timelike}
\end{equation}
\begin{equation}
\begin{aligned}h_{\mu}{}^{\alpha}h_{\nu}{}^{\beta}
\left(\dot{\omega}_{\alpha\beta}-
\dot{W}_{\left[\alpha\beta\right]}\right)=
 &
\frac{1}{2}h_{\mu}{}^{\alpha}h_{\nu}{}^{\beta}
R_{\left[\alpha\beta\right]}-E_{\left[\mu\nu\right]}-\frac{2}{3}
\theta\omega_{\mu\nu}+2\sigma^{\alpha}{}_{[\mu}
\omega_{\nu]\alpha}\\
 &
+D_{\left[\mu\right.}a_{\left.\nu\right]}+
X_{\left[\mu\right.}a_{\left.\nu\right]}+\frac{1}{3}\theta
 W_{\left[\mu\nu\right]}-
 W^{\delta}{}_{[\mu}
 \left(\sigma_{\nu]\delta}+
 \omega_{\nu]\delta}\right)
 \,,
 \end{aligned}
\label{SE_eq:omega_dot_general}
\end{equation}
\begin{equation}
\begin{aligned}
h_{\mu}{}^{\alpha}h_{\nu}{}^{\beta}
\left(\dot{\sigma}_{\alpha\beta}-\dot{W}_{\left\langle
\alpha\beta\right\rangle }\right)= & \frac{1}{2}R_{\left\langle
\mu\nu\right\rangle }-E_{\left(\mu\nu\right)}+D_{\left\langle
\mu\right.}a_{\left.\nu\right\rangle }+a_{\left\langle
\mu\right.}a_{\left.\nu\right\rangle
}-\frac{2}{3}\sigma_{\mu\nu}\theta-
\sigma^{\delta}{}_{\left\langle
\mu\right.}\sigma_{\left.\nu\right\rangle\delta }
\\ &
-
\omega_{\delta\left\langle \mu\right.}
\omega_{\left.\nu\right\rangle }{}^{\delta}
+
X_{\left\langle \mu\right.}a_{\left.\nu\right\rangle}
+
W_{\delta\left\langle \mu\right.}\sigma_{\left.\nu\right\rangle }{}^{\delta}
+
W_{\delta\left\langle \mu\right.}\omega_{\left.\nu\right\rangle }{}^{\delta}
+
\frac{1}{3}W_{\left\langle \mu\nu\right\rangle }\theta\,,
\end{aligned}
\label{SE_eq:sigma_dot_general}
\end{equation}
where for any 2-tensor $Y_{\alpha\beta}$
we use the angular brackets to represent the projected
symmetric part without trace of it, i.e., 
$Y_{\left\langle \alpha\beta\right\rangle }
\equiv\left[h^{\mu}{}_{(\alpha}h_{\beta)}{}^{\nu}-
\frac{h_{\alpha\beta}}{3}h^{\mu\nu}\right]
Y_{\mu\nu}$, and
dummy indices are leftout of all the symmetrization processes. 
Equations
\eqref{SE_eq::Raychaudhuri_timelike}--\eqref{SE_eq:sigma_dot_general}
follow from computing the projection
$h_{\mu}{}^{\alpha}u^{\beta}h_{\nu}{}^{\gamma}
R_{\alpha\beta\gamma\delta}u^{\delta}$
and evaluate, respectively, its trace, antisymmetric part and symmetric
part without trace. 
Provided the field equations of a gravity theory to relate the projection
of the Ricci tensor with the stress-energy tensor,
Eq.~(\ref{SE_eq::Raychaudhuri_timelike})
represents the generalization of the Raychaudhuri equation for manifolds
with nonzero torsion~\citep{Luz_Vitagliano_2017,Dey_Liberati_Pranzetti_2017},
describing the evolution of the expansion of a congruence of curves.
From the Ricci identity,
Eq.~(\ref{Conventions_eq:Riemann_tensor_definition}),
we also find the constraint equations,
\begin{equation}
\begin{aligned}
\varepsilon^{\alpha\beta\gamma}D_{\alpha}
\left(\omega_{\beta\gamma}-W_{\beta\gamma}\right)
-
\varepsilon^{\alpha\beta\gamma}a_{\gamma}\omega_{\alpha\beta}=
&
H_{\gamma}{}^{\gamma}
-
2\bar{S}^{\alpha\beta}\left(\frac{1}{3}h_{\beta\alpha}\theta+\sigma_{\beta\alpha}+\omega_{\beta\alpha}-W_{\beta\alpha}\right)\\
&
-\varepsilon^{\alpha\beta\gamma}a_{\gamma}\left(S_{\alpha\beta}+W_{\alpha\beta}\right)\,,
\label{SE_eq:omega_divergence_general}
\end{aligned}
\end{equation}
\begin{equation}
\begin{aligned}
\varepsilon^{\alpha\beta\left\langle \mu\right.}
D_{\alpha}\left(\sigma_{\beta}{}^{\left.\nu\right\rangle }+\omega_{\beta}{}^{\left.\nu\right\rangle }-W_{\beta}{}^{\left.\nu\right\rangle }\right)
+\varepsilon^{\alpha\beta\left\langle \mu\right.}a^{\left.\nu\right\rangle }\omega_{\beta\alpha}=
&
H^{\left\langle \mu\nu\right\rangle }
-
\varepsilon^{\alpha\beta\left\langle \mu\right.}a^{\left.\nu\right\rangle }\left(S_{\alpha\beta}+W_{\alpha\beta}\right)+
2W_{\delta}{}^{\left\langle \mu\right.}\bar{S}^{\left.\nu\right\rangle \delta}\\
&
-2\left(\frac{1}{3}h_{\delta}{}^{\left\langle \mu\right.}\theta+\sigma_{\delta}{}^{\left\langle \mu\right.}+\omega_{\delta}{}^{\left\langle \mu\right.}\right)\bar{S}^{\left.\nu\right\rangle \delta}\,,
\label{SE_eq:sigma_omega_curl_general}
\end{aligned}
\end{equation}
\begin{equation}
\begin{aligned}
\frac{2}{3}D_{\mu}\theta
-
D_{\alpha}\left(\sigma_{\mu}{}^{\alpha}+\omega_{\mu}{}^{\alpha}-W_{\mu}{}^{\alpha}\right)
-
D_{\mu}W_{\gamma}{}^{\gamma}-2a^{\gamma}\omega_{\mu\gamma}=
&
-h^{\alpha}{}_{\mu}R_{\alpha\beta}u^{\beta}
-
2\varepsilon_{\alpha\beta\mu}\bar{S}^{\alpha\gamma}W_{\gamma}{}^{\beta}
\\
+2a^{\gamma}\left(S_{\gamma\mu}+W_{\left[\gamma\mu\right]}\right)
&
+2\varepsilon_{\alpha\beta\mu}\bar{S}^{\alpha\gamma}\left(\frac{1}{3}h_{\gamma}{}^{\beta}\theta+\sigma_{\gamma}{}^{\beta}+\omega_{\gamma}{}^{\beta}\right)\,.
\label{SE_eq:theta_gradient}
\end{aligned}
\end{equation}
where dummy indices do not participate in the symetrization
processes.
Equations \eqref{SE_eq:omega_divergence_general}--\eqref{SE_eq:theta_gradient}
follow from computing the projection
$\varepsilon^{\alpha\beta\lambda}h_{\rho}{}^{\gamma}
R_{\alpha\beta\gamma\delta}u^{\delta}$
and evaluate, respectively, its trace, symmetric part without trace
and antisymmetric part.  These equations clearly exemplify how the
presence of torsion modifies the geometry of the manifold and,
consequently, the change in the evolution of a congruence of timelike
curves. When comparing to the case of vanishing
torsion~\citep{Ehlers_1961,Ellis_vanElst_1999}, we see that, in the
presence of a general torsion tensor field, the magnetic part of the
Weyl tensor, $H$, is characterized by
Eqs.~(\ref{SE_eq:omega_divergence_general})
and (\ref{SE_eq:sigma_omega_curl_general}), in
particular, it also depends on the divergence of the vorticity vector
$\frac{1}{2}\varepsilon^{\gamma\mu\nu}\omega_{\mu\nu}$. Moreover, from
Eq.~(\ref{SE_eq:omega_divergence_general}), we conclude that the
presence of torsion acts as a cause for the rotation of the congruence.

The evolution and constraint equations for the components of the
Weyl tensor are found from the identity for the Weyl tensor given in
Eq.~(\ref{Conventions_eq:Div_Weyl_trumper}), or, equivalently, from the
second Bianchi identity,
Eq.~(\ref{Conventions_eq:second_Bianchi_identity}).  For the electric part
of the Weyl tensor we find the propagation equation 
\begin{align}
&-h_{\alpha\mu}h_{\beta\nu}\dot{E}^{\mu\nu}+
\varepsilon_{\mu\beta}{}^{\nu}\left(D_{\nu}\bar{H}^{\mu}{}_{\alpha}+
a_{\nu}\bar{H}^{\mu}{}_{\alpha}\right)+
\varepsilon^{\mu}{}_{\alpha\delta}a^{\delta}H_{\mu\beta}
+\left(\sigma_{\alpha\nu}+\omega_{\alpha\nu}-
W_{\alpha\nu}\right)E^{\nu}{}_{\beta}
\nonumber\\
&+E_{\alpha}{}^{\mu}
\left(2\sigma_{\mu\beta}+2\omega_{\mu\beta}-W_{\mu\beta}\right)
-E_{\alpha\beta}\left(\theta-W_{\mu}{}^{\mu}\right)-
h_{\alpha\beta}E^{\nu\mu}\left(\sigma_{\mu\nu}+\omega_{\mu\nu}-
\frac{1}{2}W_{\mu\nu}\right)=
\nonumber\\
&=\,\frac{1}{4}h_{\alpha\beta}R_{\nu\mu}\left(W^{\mu\nu}-
W_{\gamma}{}^{\gamma}h^{\mu\nu}-2
\varepsilon^{\mu}{}_{\delta\gamma}\bar{S}^{\gamma\delta}
u^{\nu}\right)
-\varepsilon^{\mu\nu}{}_{\beta}X_{\mu}\bar{H}_{\alpha\nu}-
2\bar{S}_{\beta}{}^{\mu}\bar{H}_{\alpha\mu}+
h_{\alpha\beta}\bar{S}_{\mu\nu}\bar{H}^{\mu\nu}
\nonumber\\
&-\frac{1}{2}W_{\alpha\nu}h_{\delta\beta}R^{\nu\delta}-
\frac{1}{6}W_{\alpha\beta}R+\frac{1}{2}W_{\alpha\beta}
h_{\nu\delta}R^{\nu\delta}
+\frac{1}{2}h_{\alpha\nu}R^{\nu\delta}
\left(W_{\mu}{}^{\mu}h_{\delta\beta}-W_{\delta\beta}\right)-
\frac{1}{12}h_{\alpha\beta}\dot{R}
\nonumber\\
&+\frac{1}{2}h_{\alpha\mu}h_{\beta\nu}
u_{\delta}\nabla^{\delta}R^{\mu\nu}-\frac{1}{2}
D_{\alpha}\left(u_{\delta}R^{\delta}{}_{\beta}\right)
+u_{\lambda}R^{\lambda\sigma}
\bar{S}_{\beta}{}^{\mu}\varepsilon_{\sigma\mu\alpha}-
\frac{1}{2}R_{\mu\nu}X_{\alpha}u^{\mu}h^{\nu}{}_{\beta}
\nonumber\\
&+\frac{1}{2}\left(\frac{1}{3}h_{\alpha\mu}\theta+
\sigma_{\alpha\mu}+\omega_{\alpha\mu}\right)
R^{\mu\nu}h_{\nu\beta}\,,
\label{SE_eq:E_dot_general}
\end{align}
and the constraint
equation
\begin{align}
&D_{\mu}E_{\alpha}{}^{\mu}+
\varepsilon^{\mu\gamma\delta}\bar{H}_{\mu\alpha}
\left(\omega_{\delta\gamma}-W_{\delta\gamma}-
\frac{1}{2}S_{\delta\gamma}\right)
+\left(\sigma_{\delta\nu}+\omega_{\delta\nu}-
\frac{1}{2}W_{\delta\nu}\right)
\varepsilon^{\nu\beta\mu}h_{\alpha\beta}H_{\mu}{}^{\delta}
=\frac{1}{2}R_{\beta\gamma}\bar{S}^{\mu\beta}
\varepsilon^{\gamma}{}_{\mu\alpha}
\nonumber\\
&-R_{\mu\gamma}
\bar{S}^{\mu\beta}\varepsilon^{\gamma}{}_{\beta\alpha}-
\frac{1}{12}RX_{\alpha}+\frac{1}{2}R
\varepsilon_{\mu\beta\alpha}\bar{S}^{\mu\beta}
-\frac{1}{4}R^{\gamma\beta}u_{\beta}\left(W_{\alpha\gamma}-
X_{\alpha}u_{\gamma}\right)+\frac{1}{2}h_{\nu\alpha}
R^{\nu\beta}a_{\beta}+\frac{1}{12}D_{\alpha}R
\nonumber\\
&-\frac{1}{2}E_{\alpha\nu}X^{\nu}+\frac{1}{2}
D_{\alpha}\left(R_{\mu\nu}u^{\mu}u^{\nu}\right)+
R_{\nu\beta}W_{\alpha}{}^{\left(\nu\right.}u^{\left.\beta\right)}
+2\bar{S}^{\nu\beta}\varepsilon_{\nu\beta\mu}
E_{\alpha}{}^{\mu}+\frac{1}{2}S_{\alpha\gamma}
u^{\beta}R_{\beta}{}^{\gamma}-
\varepsilon_{\alpha\beta\nu}\bar{S}_{\mu}{}^{\beta}E^{\mu\nu}
\nonumber\\
&-\frac{1}{3}\theta R_{\nu\beta}
u^{\left(\nu\right.}h^{\left.\beta\right)}{}_{\alpha}
-R_{\nu\beta}u^{\left(\nu\right.}
\left(\sigma^{\left.\beta\right)}{}_{\alpha}-
\omega^{\left.\beta\right)}{}_{\alpha}\right)
-\frac{1}{2}h_{\alpha}{}^{\delta}u^{\gamma}
\nabla_{\gamma}\left(R_{\mu\nu}h^{\mu}{}_{\delta}
u^{\nu}\right)+\frac{1}{2}R_{\mu\nu}u^{\mu}u^{\nu}a_{\alpha}
\nonumber\\
&+\frac{1}{2}h_{\delta\alpha}R^{\delta\mu}
\left(\varepsilon_{\mu\nu\beta}\bar{S}^{\beta\nu}
-\frac{1}{2}W_{\gamma}{}^{\gamma}u_{\mu}+\frac{1}{2}X_{\mu}\right)\,,
\label{SE_eq:E_curl_general}
\end{align}
where only the upper indices enter in the symmetrization process.
{Equation (\ref{SE_eq:E_dot_general}) is found from the
projection
$h_{\mu\gamma}h_{\nu\beta}u_{\delta}
\nabla_{\lambda}C^{\gamma\delta\beta\lambda}$ and
Eq.~(\ref{SE_eq:E_curl_general})  follows from
$h_{\delta\alpha}u_{\gamma}
u_{\beta}\nabla_{\lambda}C^{\gamma\delta\beta\lambda}$.
For the magnetic part
of the Weyl tensor we find the propagation equation 
\begin{align}
&\left(2a^{\mu}E^{\nu}{}_{\left(\alpha\right.}+
D^{\mu}E^{\nu}{}_{\left(\alpha\right.}\right)
\varepsilon_{\left.\beta\right)\nu\mu}-
h^{\mu}{}_{\left(\alpha\right.}
h_{\left.\beta\right)}{}^{\nu}\dot{H}_{\mu\nu}
+\left(\sigma_{\mu\left(\alpha\right.}+
\omega_{\mu\left(\alpha\right.}\right)
H_{\left.\beta\right)}{}^{\mu}-
\left(\frac{2}{3}H_{\alpha\beta}+
\frac{1}{3}\bar{H}_{\alpha\beta}\right)\theta
\nonumber\\
&+\left(\frac{1}{3}h_{\alpha\beta}\theta-
h_{\alpha\beta}W^{\mu}{}_{\mu}-
\sigma_{\alpha\beta}+
W_{\left(\alpha\beta\right)}\right)\bar{H}^{\nu}{}_{\nu}
+2\sigma_{\mu\left(\alpha\right.}
\bar{H}^{\mu}{}_{\left.\beta\right)}-
\left(\sigma_{\mu\nu}-W_{\mu\nu}\right)
h_{\alpha\beta}\bar{H}^{\mu\nu}
+W^{\mu}{}_{\mu}\bar{H}_{\alpha\beta}
\nonumber\\
&-\bar{H}^{\mu}{}_{\left(\alpha\right.}W_{\left.\beta\right)\mu}
-W_{\mu\left(\alpha\right.}\bar{H}^{\mu}{}_{\left.\beta\right)}
=\frac{1}{2}D^{\delta}
\left(\varepsilon_{\gamma\delta\left(\alpha\right.}
h_{\left.\beta\right)\mu}R^{\gamma\mu}\right)+
E_{\mu\left(\alpha\right.}\left(2\bar{S}_{\left.\beta\right)}{}^{\mu}-
\varepsilon_{\left.\beta\right)}{}^{\mu\nu}X_{\nu}\right)
\nonumber\\
&-\frac{1}{2}u_{\mu}R^{\gamma\mu}\varepsilon_{\gamma\nu\left(\alpha\right.}
\left(\frac{1}{3}h^{\nu}{}_{\left.\beta\right)}\theta+
\sigma^{\nu}{}_{\left.\beta\right)}+\omega^{\nu}{}_{\left.
\beta\right)}-W^{\nu}{}_{\left.\beta\right)}\right)
+\frac{1}{2}R^{\gamma\mu}u_{\gamma}
h_{\mu\left(\alpha\right.}
\varepsilon_{\left.\beta\right)}{}^{\nu\delta}
\left(\omega_{\nu\delta}-W_{\nu\delta}\right)+
\frac{1}{3}R\bar{S}_{\left(\alpha\beta\right)}
\nonumber\\
&-\bar{S}_{\left(\alpha\right.}{}^{\mu}h^{\gamma}{}_{\left.\beta\right)}
R_{\mu\gamma}
+\frac{1}{2}\varepsilon_{\mu\nu}{}_{\left(\alpha\right.}
W^{\mu}{}_{\left.\beta\right)}R^{\nu\delta}u_{\delta}
+\bar{S}_{\left(\alpha\beta\right)}R_{\mu\delta}
u^{\mu}u^{\delta}-\frac{1}{2}
\varepsilon^{\mu\nu}{}_{\left(\alpha\right.}h_{\left.\beta\right)}{}^{\gamma}
X_{\mu}R_{\nu\gamma}\,,
\label{SE_eq:H1_dot_general}
\end{align}
where dummy indices are out of the symmetrization process.  This
equation describes the propagation of the $H$ component of the Weyl
tensor along the congruence.
For the magnetic part
of the Weyl tensor we find the constraint equation 
\begin{align}
&-2D_{\mu}H^{\mu}{}_{\alpha}+\frac{2}{3}
\varepsilon_{\alpha\beta\delta}E^{\beta\delta}
\theta+2\varepsilon_{\alpha\beta}{}^{\mu}
E^{\beta\delta}\sigma_{\delta\mu}
+4E^{\beta}{}_{\left(\delta\right.}
\varepsilon_{\left.\alpha\right)\beta}{}^{\mu}
\left(\omega^{\delta}{}_{\mu}-W^{\delta}{}_{\mu}\right)=
\varepsilon_{\alpha\gamma\delta}
D^{\delta}\left(R^{\gamma\beta}u_{\beta}\right)
\nonumber\\
&+\varepsilon_{\alpha\gamma\delta}
R^{\gamma}{}_{\beta}W^{\delta\beta}-
\frac{1}{3}\varepsilon_{\alpha\gamma\delta}
R^{\gamma\delta}\theta
-\varepsilon_{\alpha\gamma\delta}R^{\gamma\beta}
\left(\sigma^{\delta}{}_{\beta}+
\omega^{\delta}{}_{\beta}\right)+2\bar{S}^{\mu}{}_{\alpha}
R_{\mu\delta}u^{\delta}
-2\bar{S}_{\alpha}{}^{\beta}R_{\beta\delta}u^{\delta}-
2\varepsilon^{\mu\nu\gamma}S_{\mu\nu}E_{\gamma\alpha}
\nonumber\\
&+\varepsilon^{\mu\nu\gamma}S_{\mu\nu}
R_{\gamma\beta}h^{\beta}{}_{\alpha}-
\varepsilon^{\mu\nu}{}_{\alpha}
S_{\mu\nu}u^{\gamma}u^{\beta}R_{\gamma\beta}
-\frac{1}{3}\varepsilon^{\mu\nu}{}_{\alpha}
S_{\mu\nu}R-4\bar{S}^{\gamma\beta}
\varepsilon_{\gamma\beta\mu}H^{\mu}{}_{\alpha}\,.
\label{SE_eq:H1_divergence_general}
\end{align}
This equation
provides the divergence of $H$ on $\mathcal{V}$.
Equation~(\ref{SE_eq:H1_dot_general}) is found from
$\varepsilon_{\gamma\delta\left(\alpha\right.}
h_{\left.\beta\right)\mu}\nabla_{\lambda}C^{\gamma\delta\mu\lambda}$ 
and
Eq.~(\ref{SE_eq:H1_divergence_general}) follows from computing the
projection
$\varepsilon_{\alpha\gamma\delta}u_{\beta}
\nabla_{\lambda}C^{\gamma\delta\beta\lambda}$.
Computing the contraction
$\varepsilon_{\alpha}{}^{\gamma\delta}h^{\mu\beta}u^{\nu}
\left(R_{\gamma\delta\mu\nu}-R_{\mu\nu\gamma\delta}\right)$ 
and using 
Eq.~\eqref{Conventions_eq:Riemann_exchange_pair_indices},
we find that there is
a further
relation, one between the tensors $H$ and $\bar{H}$,
\begin{align}
&H_{\alpha}{}^{\beta}-
\bar{H}_{\alpha}{}^{\beta}+
\varepsilon_{\alpha}{}^{\mu\beta}u^{\nu}
R_{\left[\nu\mu\right]}=
-\frac{1}{2}\varepsilon_{\alpha}{}^{\mu\nu}D^{\beta}
\left(W_{\nu\mu}+S_{\nu\mu}\right)+
\varepsilon_{\alpha\mu\nu}X^{\nu}
B_{\perp}^{\beta\mu}
-\varepsilon_{\alpha\mu\nu}X^{\nu}
\left(W^{\left[\beta\mu\right]}+S^{\beta\mu}\right)
\nonumber\\
&-\frac{1}{2}\varepsilon_{\alpha}{}^{\mu\nu}a^{\beta}
\left(W_{\nu\mu}+S_{\nu\mu}\right)
+2h_{\alpha}{}^{\gamma}h^{\mu\beta}
u^{\nu}\nabla_{\nu}\bar{S}_{\gamma\mu}+
2\bar{S}_{\alpha}{}^{\beta}\theta-
2h_{\alpha}{}^{\beta}\bar{S}^{\mu\nu}
\left(S_{\mu\nu}+W_{\left[\mu\nu\right]}\right)
\nonumber\\
&+\varepsilon_{\alpha\mu\nu}a^{\nu}
\left(W^{\left(\mu\beta\right)}+
S^{\mu\beta}\right)-
h_{\alpha}{}^{\beta}u^{\nu}
\nabla_{\nu}\bar{S}_{\mu}{}^{\mu}-
h_{\alpha}{}^{\beta}\bar{S}_{\mu}{}^{\mu}\theta
+\varepsilon_{\alpha\mu}{}^{\nu}
D_{\nu}\left(W^{\left(\beta\mu\right)}+
S^{\beta\mu}\right)+
\varepsilon_{\alpha}{}^{\mu\nu}X^{\beta}
\omega_{\mu\nu}
\nonumber\\
&-2\bar{S}^{\mu\beta}B_{\perp\mu\alpha}
-\frac{1}{2}\varepsilon_{\alpha}{}^{\mu\nu}
X^{\beta}\left(W_{\mu\nu}+S_{\mu\nu}\right)+
2h_{\alpha}{}^{\beta}\bar{S}^{\mu\nu}
B_{\perp\mu\nu}-
2\bar{S}_{\alpha}{}^{\mu}W_{\mu}{}^{\beta}\,.
\label{SE_eq:H1_H2_relation_general}
\end{align}
Note that
in Eq.~(\ref{SE_eq:H1_H2_relation_general}) the term with the Ricci
tensor on the left-hand side could be removed by taking the symmetric
part in the indices $\alpha$ and $\beta$, however, this will add more
terms and dense the notation on the right-hand side so, we opted to
write the result as is.
Note also that 
Eq.~(\ref{SE_eq:H1_H2_relation_general})
shows how the presence of torsion is responsible for the degeneracy
removal of the magnetic parts of the Weyl tensor. It is interesting to
note that the difference between the magnetic parts of the Weyl
tensor depends on the derivatives of the components of the torsion
tensor, on $\mathcal{V}$ and along $u$, making it clear that in
general both the value and the rate of change of the torsion field
affect the difference between the tensors $H$ and $\bar{H}$.
Moreover, since in Eq.~(\ref{SE_eq:H1_H2_relation_general}) we have an
algebraic relation for the difference of the components of the tensors
$H$ and $\bar{H}$,
Eqs.~\eqref{SE_eq:H1_dot_general}--\eqref{SE_eq:H1_H2_relation_general}
characterize both $H$ and
$\bar{H}$, that is, we do not need to find propagation and constraint
equations for $\bar{H}$ since those will not be independent of
Eqs.~\eqref{SE_eq:H1_dot_general}--\eqref{SE_eq:H1_H2_relation_general}.
Using 
Eq.~(\ref{Conventions_eq:first_Bianchi_identity}),
we find
the remaining equations that characterize the torsion
tensor components. These equations
are
\begin{equation}
\begin{aligned}
h^{\alpha\nu}u^{\mu}R_{\left[\nu\mu\right]}
&=
\varepsilon^{\alpha}{}_{\mu\nu}u^{\gamma}\nabla_{\gamma}\bar{S}^{\mu\nu}
+\left(\varepsilon_{\gamma\mu\nu}\bar{S}^{\mu\nu}-\frac{1}{2}X_{\gamma}\right)B_{\perp}^{\alpha\gamma}
+\varepsilon_{\gamma}{}^{\alpha\mu}\bar{S}_{\mu\beta}\left(B_{\perp}^{\beta\gamma}-W^{\beta\gamma}\right)+\frac{1}{2}D^{\alpha}W_{\mu}{}^{\mu}\\
&
-\frac{1}{2}D_{\beta}W^{\alpha\beta}+\frac{1}{2}W_{\mu}{}^{\mu}a^{\alpha}-\frac{1}{2}W^{\alpha\gamma}a_{\gamma}+\frac{1}{2}X^{\alpha}\theta+S_{\gamma}{}^{\alpha}a^{\gamma}\,,
\label{SE_eq:W_Sbar_dot_general}
\end{aligned}
\end{equation}
and
\begin{equation}
\begin{aligned}
\varepsilon^{\alpha\mu\nu}R_{\mu\nu}
&=
\varepsilon^{\alpha}{}_{\sigma\rho}\dot{S}^{\sigma\rho}
+2D_{\beta}\bar{S}^{\beta\alpha}
+\varepsilon^{\alpha}{}_{\sigma\rho}D^{\sigma}X^{\rho}
-\varepsilon^{\alpha}{}_{\sigma\rho}W_{\mu}{}^{\mu}\left(\omega^{\sigma\rho}-W^{\sigma\rho}\right)
+2\bar{S}^{\alpha\mu}\left(X_{\mu}+a_{\mu}\right)\\
&
-\varepsilon^{\alpha}{}_{\sigma\rho}W^{\rho}{}_{\beta}\left(B_{\perp}^{\beta\sigma}-W^{\beta\sigma}\right)
-\varepsilon^{\alpha}{}_{\sigma\rho}\left(X^{\sigma}a^{\rho}+S^{\rho\sigma}\theta\right)
-4\varepsilon^{\beta\mu\nu}\bar{S}^{\alpha}{}_{\beta}\bar{S}_{\mu\nu}\,.
\label{SE_eq:S_dot_X_Sbar_curl}
\end{aligned}
\end{equation}
Equations (\ref{SE_eq:W_Sbar_dot_general}) and
(\ref{SE_eq:S_dot_X_Sbar_curl}) are derived from
computing the projections
$h^{\sigma\alpha}u^{\gamma}R_{\left[\alpha\beta\gamma\right]}{}^{\beta}$
and $h^{\sigma\alpha}h^{\rho\gamma}
R_{\left[\alpha\beta\gamma\right]}{}^{\beta}$, respectively,
and using 
the first Bianchi
identity,
Eq.~\eqref{Conventions_eq:first_Bianchi_identity}.

Equations
\eqref{SE_eq::Raychaudhuri_timelike}--\eqref{SE_eq:S_dot_X_Sbar_curl}
characterize the geometry of the manifold, containing exactly the
same information as the Ricci and Bianchi identities.

\section{The stress-energy-momentum tensor and the structure equations
for the matter fields
\label{Section:EC_stressenergytensor}}

\subsection{The stress-energy tensor and its decomposition}

For the stress-energy tensor $T$, that characterizes the matter fields
permeating the spacetime manifold, we also want to apply the 1+3
formalism in order to study its dynamical evolution. Setting the
congruence's tangent vector field $u$ to coincide with the 4-velocity
of an observer, without imposing any symmetries on ${T}$ and using
Eq.~(\ref{Projector_eq:orthogonal_projector_timelike}) we find
the following decomposition
\begin{equation}
{T}_{\alpha\beta}=\mu\,u_{\alpha}u_{\beta}+p\,h_{\alpha\beta}+
q_{1\alpha}u_{\beta}+u_{\alpha}q_{2\beta}+\pi_{\alpha\beta}+
\varepsilon_{\alpha\beta}{}^{\gamma}m_{\gamma}\,
,\label{SWT_eq:Stress_energy_tensor_1p3_decomposition}
\end{equation}
with
\begin{equation}
\begin{aligned}
\mu&=u^{\mu}u^{\nu} T_{\mu\nu}\,,&p&=\frac{1}{3}h^{\mu\nu} T_{\mu\nu}\,,&q_{1\alpha}&=-h_{\alpha}{}^{\mu}u^{\nu} T_{\mu\nu}\,,\\q_{2\alpha}&=-u^{\mu}h_{\alpha}{}^{\nu} T_{\mu\nu}\,,&\pi_{\alpha\beta}&= T_{\left\langle \alpha\beta\right\rangle }\,,&m_{\alpha}&=\frac{1}{2}\varepsilon_{\alpha}{}^{\mu\nu} T_{\mu\nu}\,.
\label{SWT_eq:Stress_energy_tensor_1p3_decomposition_quantities}
\end{aligned}
\end{equation}
where $\mu$ is the energy density measured by the chosen observer,
$p$ is the pressure, $q_{1\alpha}$  and 
$q_{2\alpha}$ represent energy and momentum density fluxes,
$\pi_{\alpha\beta}$ is the anysotropic stress and $m_{\alpha} $ is a
flux, in particular related with the nonconservation of intrinsic
angular momentum of matter.

Note that we are free to arbitrarily choose the time-like congruence,
nonetheless, in the case of a single fluid, it is useful to set the
congruence's tangent vector field $u$ to coincide with the 4-velocity
of the elements of volume of the fluid, in which case the various
projections of the stress-energy tensor and the kinematical quantities
of the congruence directly represent the properties and evolution of
the matter fluid.

\subsection{The structure equations for the matter fields}

To find the set of equations describing the dynamical evolution of the
matter fields in the manifold, we will consider the general
conservation law for the stress-energy tensor, given in general
by
\begin{equation}
\nabla_{\beta}{T}^{\alpha\beta}=\Psi^{\alpha}\,,
\label{SE_eq:Div_stress_energy_general}
\end{equation}
where $\Psi^{\alpha}$ is some tensor to be determined by the field
equations and the Bianchi identities. From
Eqs.~\eqref{SWT_eq:Stress_energy_tensor_1p3_decomposition} and
\eqref{SE_eq:Div_stress_energy_general}, the projections along $u$ and
on $\mathcal{V}$ are
\begin{align}
\dot{\mu}+\left(\theta-W_{\alpha}{}^{\alpha}\right)
\left(\mu+p\right)-\varepsilon^{\alpha\beta\gamma}
m_{\gamma}\left(\omega_{\alpha\beta}-
W_{\alpha\beta}\right)
+\pi^{\alpha\beta}\left(\sigma_{\alpha\beta}-
W_{\alpha\beta}\right)+
\left(q_{1}^{\alpha}+q_{2}^{\alpha}\right)a_{\alpha}+
D_{\alpha}q_{2}^{\alpha}
=-u_{\alpha}\Psi^{\alpha}\,,
\label{SE_eq:energy_conservation_stress_energy_general}
\end{align}
\begin{equation}
\begin{aligned}
\left(\mu+p\right)a_{\alpha}+D_{\alpha}p+D_{\mu}\pi_{\alpha}{}^{\mu}+\varepsilon_{\alpha}{}^{\mu\nu}D_{\mu}m_{\nu}+\left(\pi_{\alpha\nu}-\varepsilon_{\alpha\mu\nu}m^{\mu}\right)a^{\nu}+h_{\alpha}{}^{\beta}\dot{q}_{1\beta}+\left(q_{1\alpha}+\frac{1}{3}q_{2\alpha}\right)\theta&\\-q_{1\alpha}W_{\beta}{}^{\beta}+q_{2}^{\beta}\left(\sigma_{\beta\alpha}+\omega_{\beta\alpha}-W_{\beta\alpha}\right)&=h_{\alpha\beta}\Psi^{\beta}\,.
\label{SE_eq:momentum_conservation_stress_energy_general}
\end{aligned}
\end{equation}
At this point the imposition that the evolution equations
for the matter variables are determined by
\eqref{SE_eq:Div_stress_energy_general} is given ad hoc. In
practice, however, provided the field equations of a gravity theory
relating the Ricci and the stress-energy tensors, the conservation
equations will follow from the second Bianchi identity.  Hence, these
are a pivotal component to guarantee the consistency of the physical
theory and system of equations.

\section{The Einstein-Cartan theory for a Weyssenhoff like
torsion: Field equations
\label{Section:EC_structure_equations}}

The general set of structure equations that arise from the 1+3
formalism can be used to study solutions of any relativistic theory of
gravitation based on an affine, metric compatible connection. In this
section, we will focus on the Einstein-Cartan theory characterized by
the field equations
\begin{align}
R_{\alpha\beta}-\frac{1}{2}g_{\alpha\beta}R  +
\Lambda g_{\alpha\beta} & =\hskip0.3cm 8\pi{T}_{\alpha\beta}\,,
\label{EC_eq:field_equations_Ricci}\\
S^{\alpha\beta\gamma}+
2g^{\gamma[\alpha}S^{\beta]}{}_{\mu}{}^{\mu} & =
-8\pi\Delta^{\alpha\beta\gamma}\,,
\label{EC_eq:field_equations_torsion}
\end{align}
where ${T}_{\alpha\beta}$ represents the canonical stress-energy
tensor, $\Delta^{\alpha\beta\mu}$ is the intrinsic hypermomentum and
$\Lambda$ the cosmological constant
The  Einstein-Cartan theory defined by Eqs.~\eqref{EC_eq:field_equations_Ricci}
and Eq.~\eqref{EC_eq:field_equations_torsion}
can be derived from the Einstein-Hilbert
action
$I=\frac{1}{16\pi}\int{ d^4x\sqrt{-g}\left(R-2\Lambda\right)}
+\int{ d^4x\sqrt{-g}{\cal L}_{\rm m}}$,
where
the Ricci scalar contains the metric and the torsion as dynamical
variables, ${\cal L}_{\rm m}$ is the matter Lagrangian density, 
and the variation of $I$ must be performed with respect to those
two fields.
The conservation law is given by
\begin{equation}
\nabla_{\beta}{T}_{\alpha}{}^{\beta}=
2S_{\alpha\mu\nu}{T}^{\nu\mu}-\frac{1}{4\pi}
S_{\alpha\mu}{}^{\mu}\Lambda+\frac{1}{8\pi}
\left(S_{\alpha\mu}{}^{\mu}R-S^{\mu\nu\sigma}
R_{\alpha\sigma\mu\nu}\right)\,.
\label{EC_eq:Conservation_laws}
\end{equation}

To simplify the equations and, in agreement with what we are going to
consider in the following, we will impose that the torsion tensor
is characterized only by the tensor $S_{\alpha\beta}$, that is, the
tensors $\bar{S}_{\alpha\beta}$, $W_{\alpha\beta}$ and $X_{\alpha}$ in
Eq.~(\ref{SWT_eq:timelike_torsion_decomposition}) are considered to be
identically zero, so
\begin{equation}
S_{\alpha\beta}{}^{\gamma}=S_{\alpha\beta}u^{\gamma}\,.
\label{EC_eq:simplifiedtorsion}
\end{equation}
Given
Eqs.~\eqref{EC_eq:field_equations_Ricci}--\eqref{EC_eq:Conservation_laws}
and assuming Eq.~\eqref{EC_eq:simplifiedtorsion},
the 1+3 structure equations have the following new forms.

The propagation equations for the kinematical quantities associated
with $u$ are
\begin{align}
\dot{\theta}=  -4\pi\left(\mu+3p\right)+
\Lambda-\left(\frac{1}{3}\theta^{2}+\sigma_{\alpha\beta}
\sigma^{\alpha\beta}+\omega_{\alpha\beta}\omega^{\beta\alpha}\right)+
D_{\alpha}a^{\alpha}+a_{\mu}a^{\mu}\,
\label{EC_eq:Raychaudhuri_general}
\end{align}
\begin{align}
h_{\mu}{}^{\alpha}h_{\nu}{}^{\beta}\dot{\omega}_{\alpha\beta}=
 -E_{\left[\mu\nu\right]}+4\pi\varepsilon_{\mu\nu\gamma}m^{\gamma}-
\frac{2}{3}\theta\omega_{\mu\nu}
+2\sigma^{\alpha}{}_{\left[\mu\right.}
\omega_{\left.\nu\right]\alpha}+D_{\left[\mu\right.}
a_{\left.\nu\right]}\,,
\label{EC_eq:omega_dot_general}
\end{align}
\begin{align}
h_{\mu}{}^{\alpha}h_{\nu}{}^{\beta}\dot{\sigma}_{\alpha\beta}=
 -E_{\left(\mu\nu\right)}+4\pi\left(\pi_{\mu\nu}\right)+
D_{\left\langle \mu\right.}a_{\left.\nu\right\rangle }+
a_{\left\langle \mu\right.}a_{\left.\nu\right\rangle }-
\frac{2}{3}\sigma_{\mu\nu}\theta
-\sigma^{\delta}{}_{\left\langle \mu\right.}
\sigma_{\left.\nu\right\rangle \delta}
-\omega^{\delta}{}_{\left\langle \mu\right.}
\omega_{\left.\nu\right\rangle \delta}\,,
\label{EC_eq:sigma_dot_general}
\end{align}
and the corresponding constraint equations are
\begin{align}
\varepsilon^{\mu\nu\rho}
D_{\mu}\omega_{\nu\rho}+\varepsilon^{\mu\nu\rho}
a_{\rho}\omega_{\nu\mu} & =H_{\rho}{}^{\rho}+
\varepsilon^{\mu\nu\rho}a_{\rho}S_{\nu\mu}\,,
\label{EC_eq:EC_omega_divergence_general}
\end{align}
\begin{align}
\varepsilon^{\alpha\beta\left\langle
\mu\right.}D_{\alpha}\left(\sigma_{\beta}{}^{\left.\nu\right\rangle }+
\omega_{\beta}{}^{\left.\nu\right\rangle }\right)+
\varepsilon^{\alpha\beta\left\langle \mu\right.}
a^{\left.\nu\right\rangle }\omega_{\beta\alpha}  =
H^{\left\langle \mu\nu\right\rangle }-
\varepsilon^{\alpha\beta\left\langle
\mu\right.}a^{\left.\nu\right\rangle }
S_{\alpha\beta}\,,
\label{EC_eq:sigma_omega_curl_general}
\end{align}
\begin{align}
\frac{2}{3}D_{\alpha}\theta-D_{\mu}
\left(\sigma_{\alpha}{}^{\mu}+
\omega_{\alpha}{}^{\mu}\right)-2a^{\mu}\omega_{\alpha\mu}
=8\pi q_{1\alpha}+2a^{\mu}S_{\mu\alpha}\,.
\label{EC_eq:theta_gradient_general}
\end{align}
where only upper indices enter the symetrization
process.

The propagation equations for the electric and magnetic parts of the
Weyl tensor are
\begin{align}
&-h_{\alpha\mu}h_{\beta\nu}\dot{E}^{\mu\nu}+
\varepsilon_{\mu\beta}{}^{\nu}\left(D_{\nu}\bar{H}^{\mu}{}_{\alpha}
+a_{\nu}\bar{H}^{\mu}{}_{\alpha}\right)+
\varepsilon^{\mu}{}_{\alpha\delta}a^{\delta}H_{\mu\beta}
+\left(\sigma_{\alpha\nu}+\omega_{\alpha\nu}\right)E^{\nu}{}_{\beta}+
2E_{\alpha}{}^{\mu}\left(\sigma_{\mu\beta}+\omega_{\mu\beta}\right)
\nonumber
\\
&-E_{\alpha\beta}\theta-h_{\alpha\beta}E^{\nu\mu}
\left(\sigma_{\mu\nu}+\omega_{\mu\nu}\right)=
\frac{4\pi}{3}h_{\alpha\beta}\dot{\mu}+4\pi
h_{\alpha\mu}h_{\beta\nu}\dot{\pi}^{\mu\nu}+
4\pi\varepsilon_{\alpha\beta}{}^{\gamma}\dot{m}_{\gamma}+
4\pi\left(q_{1\alpha}a_{\beta}+a_{\alpha}q_{2\beta}\right)
\nonumber
\\
&+4\pi
D_{\alpha}q_{2\beta}
+4\pi\left(\frac{1}{3}h_{\alpha\delta}\theta+
\sigma_{\alpha\delta}+\omega_{\alpha\delta}\right)
\left[h^{\delta}{}_{\beta}\left(\mu+p\right)+
\pi^{\delta}{}_{\beta}+
\varepsilon^{\delta}{}_{\beta\gamma}m^{\gamma}\right]\,,
\label{EC_eq:E_dot_general} 
\end{align}
\begin{align}
&\bar{H}_{\mu}{}^{\mu}\left(\frac{1}{3}
h^{\alpha\beta}\theta-\sigma^{\alpha\beta}\right)+
2\bar{H}_{\mu}{}^{\left(\alpha\right.}
\sigma^{\left.\beta\right)\mu}-h^{\alpha\beta}
\bar{H}^{\mu\nu}\sigma_{\mu\nu}
+\left[2a_{\mu}E_{\nu}{}^{\left(\alpha\right.}+
D_{\mu}E_{\nu}{}^{\left(\alpha\right.}\right]
\varepsilon^{\left.\beta\right)\nu\mu}-
h^{\mu\left(\alpha\right.}h^{\left.\beta\right)\nu}\dot{H}_{\mu\nu}
\nonumber
\\
&+H^{\mu\left(\alpha\right.}
\left(\sigma_{\mu}{}^{\left.\beta\right)}+
\omega_{\mu}{}^{\left.\beta\right)}\right)-
\frac{1}{3}\left(2H^{\alpha\beta}+\bar{H}^{\alpha\beta}\right)\theta=
4\pi\varepsilon_{\gamma}{}^{\delta
\left(\alpha\right.}h^{\left.\beta\right)\mu}
q_{1}^{\gamma}\left(\sigma_{\delta\mu}+\omega_{\delta\mu}\right)+
4\pi\varepsilon^{\mu\delta(\alpha}q_{2}^{\beta)}\omega_{\delta\mu}
\nonumber
\\
&+4\pi\varepsilon_{\gamma}{}^{\delta\left(\alpha\right.}
D_{\delta}\pi^{\beta)\gamma}+4\pi
D^{(\alpha}m^{\beta)}-4\pi h^{\alpha\beta}D_{\delta}m^{\delta}\,,
\label{EC_eq:H1_dot_general}
\end{align}
and the corresponding constraint equations are,
\begin{align}
&D_{\beta}E_{\alpha}{}^{\beta}+
\varepsilon^{\beta\gamma\delta}\bar{H}_{\beta\alpha}
\left(\omega_{\delta\gamma}-\frac{1}{2}S_{\delta\gamma}\right)
+\left(\sigma_{\delta\nu}+\omega_{\delta\nu}\right)
\varepsilon^{\nu\beta\gamma}h_{\alpha\beta}H_{\gamma}{}^{\delta}=
4\pi
D_{\alpha}p+
\frac{8\pi}{3}D_{\alpha}\mu
\nonumber
\\
&+4\pi
\left[\pi_{\alpha}{}^{\beta}-
\varepsilon^{\beta}{}_{\alpha\gamma}m^{\gamma}\right]a_{\beta}
+4\pi\left(q_{2\lambda}+
q_{1\lambda}\right)
\left(\sigma_{\alpha}{}^{\lambda}+
\omega_{\alpha}{}^{\lambda}+\frac{1}{3}
h_{\alpha}{}^{\lambda}\theta\right)
+4\pi
h_{\alpha}{}^{\gamma}\dot{q}_{1\gamma}
\nonumber
\\
&+4\pi\left(\mu+p\right)a_{\alpha}-4\pi
S_{\alpha\gamma}q_{2}^{\gamma}\,,
\label{EC_eq:E_curl_general}
\end{align}
\begin{align}
&4
{E^{\beta}}_{\left(\delta\right.}
\varepsilon_{\left.\alpha\right)\beta}{}^{\gamma}
{\omega^{\delta}}_{\gamma}
+
2\varepsilon_{\alpha\beta}{}^{\gamma}
E^{\beta\delta}\sigma_{\delta\gamma}
-2D_{\gamma}H^{\gamma}{}_{\alpha}+\frac{2}{3}
\varepsilon_{\alpha\beta\delta}E^{\beta\delta}\theta=
 -8\pi\varepsilon_{\alpha\gamma\delta}\left[D^{\delta}
q_{1}^{\gamma}+\omega^{\delta\gamma}\left(\mu+p\right)\right]
\nonumber
\\
&-8\pi\varepsilon_{\alpha\gamma\delta
 }\left(\pi^{\gamma\beta}+
 \varepsilon^{\gamma\beta}{}_{\nu}m^{\nu}\right)
 \left(\sigma^{\delta}{}_{\beta}+\omega^{\delta}{}_{\beta}\right)
 -\frac{16\pi}{3}\theta m_{\alpha}-\frac{8\pi}{3}
 \varepsilon^{\mu\nu}{}_{\alpha}S_{\mu\nu}
 \left(\mu+3p-\frac{\Lambda}{4\pi}\right)
\nonumber
\\
 &+8\pi\varepsilon^{\mu\nu\gamma}S_{\mu\nu}
 \left(\pi_{\gamma\alpha}+
 \varepsilon_{\gamma\alpha\nu}m^{\nu}\right)-
 2\varepsilon^{\sigma\nu\gamma}S_{\sigma\nu}E_{\gamma\alpha}\,,
\label{EC_eq:H1_divergence_general}
\end{align}
\begin{align}
H_{\alpha}{}^{\beta}-
\bar{H}_{\alpha}{}^{\beta}+4\pi
\varepsilon_{\alpha}{}^{\mu\beta}\left(q_{1\mu}-q_{2\mu}\right)=
 -\frac{1}{2}\varepsilon_{\alpha}{}^{\mu\nu}D^{\beta}S_{\nu\mu}-
\frac{1}{2}\varepsilon_{\alpha}{}^{\mu\nu}a^{\beta}S_{\nu\mu}
  +\varepsilon_{\alpha\mu\nu}
 a^{\nu}S^{\mu\beta}+
 \varepsilon_{\alpha\mu}{}^{\nu}D_{\nu}S^{\beta\mu}\,.
\label{EC_eq:H1_H2_relation_general}
\end{align}

The equations that characterize the torsion tensor are,
\begin{equation}
4\pi\left(q_{2}^{\alpha}-q_{1}^{\alpha}\right)=
S_{\gamma}{}^{\alpha}a^{\gamma}\,,
\label{EC_eq:heat_flow_general}
\end{equation}
\begin{equation}
16\pi m^{\alpha}= \varepsilon^{\alpha}{}_{\rho\sigma}
\left(S^{\rho\sigma}\theta+\dot{S}^{\rho\sigma}\right)\,.
\label{EC_eq:spin_conservation_general}
\end{equation}

The equation relating the torsion to the
hypermomentum is 
\begin{align}
S^{\alpha\beta}u^{\gamma}
= -8\pi\Delta^{\alpha\beta\gamma}\,.
\label{EC_eq:torsionhyperm}
\end{align}

The conservation of energy and momentum equations are
\begin{align}
\dot{\mu}+\theta\left(\mu+p\right)+
2q_{1}^{\alpha}a_{\alpha}+
D_{\alpha}q_{2}^{\alpha}+\pi^{\alpha\beta}
\sigma_{\alpha\beta}+\varepsilon^{\alpha\beta\gamma}
m_{\gamma}\omega_{\beta\alpha}=0\,,
\label{EC_eq:energy_conservation_stress_energy_general}
\end{align}
\begin{equation}
\begin{aligned}
\left(\mu+p\right)a_{\alpha}
+D_{\alpha}p+h_{\alpha}{}^{\beta}\dot{q}_{1\beta}
+D_{\mu}\pi_{\alpha}{}^{\mu}
+\left(\pi_{\alpha\nu}-\varepsilon_{\alpha\mu\nu}m^{\mu}\right)a^{\nu}
&\\+\left(q_{1\alpha}+\frac{q_{2\alpha}}{3}\right)\theta+q_{2}^{\beta}\left(\sigma_{\beta\alpha}+\omega_{\beta\alpha}\right)+\varepsilon_{\alpha}{}^{\mu\nu}D_{\mu}m_{\nu}
&
=-\frac{1}{8\pi}\bar{H}_{\alpha}{}^{\rho}S^{\gamma\delta}\varepsilon_{\rho\gamma\delta}-S_{\alpha}{}^{\beta}q_{2\beta}\,.
\label{EC_eq:momentum_conservation_stress_energy_general}
\end{aligned}
\end{equation}
Once the matter model is given,
Eqs.~\eqref{EC_eq:Raychaudhuri_general}--\eqref{EC_eq:momentum_conservation_stress_energy_general} completely
describe the geometry of the spacetime and the evolution of the matter
fluid for the Einstein-Cartan theory, for a torsion tensor of the
form given in Eq.~\eqref{EC_eq:simplifiedtorsion}, i.e.,
$S_{\alpha\beta}{}^{\gamma}=S_{\alpha\beta}u^{\gamma}$.  Note,
however, that we have not yet imposed any restriction on the
stress-energy tensor, and for a torsion that assumes the form of
Eq.~\eqref{EC_eq:simplifiedtorsion}, the field equations,
Eqs.~\eqref{EC_eq:Raychaudhuri_general}--\eqref{EC_eq:momentum_conservation_stress_energy_general}, are valid
for any matter model.

The form of the field
equations,
Eqs.~\eqref{EC_eq:Raychaudhuri_general}--\eqref{EC_eq:momentum_conservation_stress_energy_general}
allow us to compare them
with the results in the literature and test their validity.
First, we see that our results differ from the ones
in~\citep{Brechet_Hobson_Lasenby_(2007)}. In this reference the
authors seem to have not
realized that in the presence of torsion, the Weyl
tensor is characterized by three tensors,
more specifically, the magnetic part
of the Weyl tensor is described by two distinct tensors; moreover,
it is quite surprising that the authors did not verify that the electric
and magnetic parts of the Weyl tensor do not carry all the usual
symmetries found in spacetimes with vanishing torsion.
Second, setting
the torsion terms in
Eqs.~\eqref{EC_eq:Raychaudhuri_general}--\eqref{EC_eq:momentum_conservation_stress_energy_general} 
to zero, $\bar{H}=H$ and both the electric and magnetic part
of the Weyl tensor are symmetric, tracefree tensors,
and imposing the stress-energy tensor to be symmetric,
such that $m^{\alpha}=0$
and $q_{1\alpha}=q_{2\alpha}$, we recover the expressions for the
structure equations for the theory of general
relativity~\citep{Ehlers_1961,Ellis_vanElst_1999}.

\section{Relativistic cosmology in Einstein-Cartan theory:
The isotropic universe and the geometry of the 3-spaces}
\label{Section:Isotropic_universe}

\subsection{Field equations for the universe with
homogeneous spinning fluid}

The general set of structure equations for the Einstein-Cartan theory,
Eqs.~\eqref{EC_eq:Raychaudhuri_general}--\eqref{EC_eq:momentum_conservation_stress_energy_general},
even for a simplified torsion tensor, is extremely complicated and
to find nontrivial solutions we have to impose some idealized symmetries
and constraints on the matter fields.
As a particular application of the previous set of equations,
we will consider the effects of a neutral Weyssenhoff fluid, see, e.g.,
\cite{Ray_Smalley_1983}, in a cosmological setting.

The Weyssenhoff fluid represents a semi-classical model for a perfect
fluid composed by fermions, taking into account the macroscopic effects
of the intrinsic angular momentum of its constituents.
Following,
Refs.~\citep{Ray_Smalley_1983,Obukhov_Korotky_1987}, for a comoving observer,
the canonical
stress-energy tensor of a Weyssenhoff fluid is such that
${T}={T}\left(\mu,p,q_{1}\right)$,
that is, the canonical stress-energy tensor only depends on the energy
density, pressure and an heat flow term that arises from the intrinsic
spin of the particles.
For the Weyssenhoff fluid the intrinsic
hypermomentum can be written as
$
\Delta^{\alpha\beta\gamma}=-\frac{1}{8\pi}\Delta^{\alpha\beta}u^{\gamma}\,,
$
where $u$ represents the proper 4-velocity of an element of volume
of the fluid and the antisymmetric spin density tensor,
$\Delta^{\alpha\beta}$,
verifies $\Delta^{\alpha\beta}u_{\beta}=0$. From
the field equation~(\ref{EC_eq:torsionhyperm}),
we find that the torsion tensor is given by
$S^{\alpha\beta}=\Delta^{\alpha\beta}$, 
and the components $\bar{S}^{\alpha\beta}$, $W^{\alpha\beta}$ and $X^{\alpha}$,
are identically zero
for the Weyssenhoff fluid.
An interesting consequence
for the Weyssenhoff model is given by Eq.~\eqref{EC_eq:heat_flow_general},
which simplifies to
$q_{1}^{\alpha}
=-\frac{1}{4\pi}S_{\gamma}{}^{\alpha}a^{\gamma}$.
This relation between the
vector field $q_{1}$ and the torsion tensor was already found by
Obukhov and Korotky for the Weyssenhoff fluid stress-energy
tensor~\citep{Obukhov_Korotky_1987}.
Of course the model found in~\citep{Obukhov_Korotky_1987} is more
general, since it is independent of the considered gravitational theory,
 showing, nonetheless, the consistency of the results.

We are interested in studying solutions where
a neutral Weyssenhoff
fluid acts as a source
of spin and that could be used to model the universe at very large scales,
such that the cosmological principle is verified by the matter fluid.
So, for the cosmological model we further assume a number of conditions.
(i) The shear tensor field of the fluid is
identically zero at every point and throughout the fluid's evolution,
hence $\sigma_{\alpha\beta}=0$.
(ii) There are no spatial expansion gradients, such that
$D_{\alpha}\theta=0$.
(iii) The matter fluid has no intrinsic preferred spatial directions,
therefore we impose that there are no spatial energy
 density and pressure gradients, namely, $D_{\alpha}\mu=0$
 and $D_{\alpha}p=0$.
(iv) The fluid's elements of volume have zero 4-acceleration
at all points and throughout the fluid's evolution, $a^{\mu}=0$.
(v) The vorticity
tensor is such that $\omega_{\alpha\beta}=S_{\alpha\beta}$. This
constraint is equivalent to impose that the spatial spaces, orthogonal
to the curves of the congruence, are hypersurfaces~\citep{Luz_Mena_2020}.
(vi) The orthogonal spatial hypersurfaces are complete and simply-connected.

As we will see, these conditions guarantee that at the level of the
metric there are no preferred spacial directions. On the other hand,
an observer comoving with the fluid that interacts directly with
the torsion tensor will in fact measure a preferred spacial direction,
however this does not imply an intrinsic anisotropy of the matter fluid.
We will discuss
this in more detail 
below.

In what follows, it is useful to define the vector fields
\begin{align}
\omega^{\gamma} 
=\frac{1}{2}\varepsilon^{\gamma\mu\nu}\omega_{\mu\nu}\,, 
\qquad
{S}^{\gamma} 
=\frac{1}{2}\varepsilon^{\gamma\mu\nu}S_{\mu\nu}\,,
\qquad
{\Delta}^{\gamma} 
=\frac{1}{2}\varepsilon^{\gamma\mu\nu}\Delta_{\mu\nu}\,,
\qquad{E}^{\gamma}  =\frac{1}{2}\varepsilon^{\gamma\mu\nu}E_{\mu\nu}\,,
\end{align}
such that $\omega_{\mu\nu}=\varepsilon_{\mu\nu\gamma}\omega^{\gamma}$,
$\Delta_{\mu\nu}=\varepsilon_{\mu\nu\gamma}{\Delta}^{\gamma}$,
$S_{\mu\nu}=\varepsilon_{\mu\nu\gamma}{S}^{\gamma}$, and
$E_{\left[\mu\nu\right]}=\varepsilon_{\mu\nu\gamma}{E}^{\gamma}$.
Then, the structure
equations
\eqref{EC_eq:Raychaudhuri_general}--\eqref{EC_eq:momentum_conservation_stress_energy_general}, together
with the previous assumptions yield the 
following set of equations.

We have for the kinematical quantities
\begin{align}
&\dot{\theta}  =-4\pi\left(\mu+3p\right)+\Lambda-
\left(\frac{1}{3}\theta^{2}-2S^{\sigma}S_{\sigma}\right)\,,
\label{Cosmology_eq:Raychaudhuri_equation}\\
&D_{\alpha}\theta  =0
\label{Dalpha=0}\,,\\
&\omega^{\gamma}  =S^{\gamma}
\label{omegagammaSgamma}
\,,\\
&\sigma_{\alpha\beta}  =0,
\label{sigmaab=0}\\
&a^{\beta} =0
\label{abeta=0}\,;
\end{align}

for the Weyl tensor components
\begin{align}
&{E}^{\gamma}  =\frac{1}{3}\theta
S^{\gamma}\,,
\label{Cosmology_eq:E_antisymmetric__torsion}\\
&E_{\left(\mu\nu\right)}  =-S_{\left\langle
\mu\right.}S_{\left.\nu\right\rangle
}\,,
\label{Cosmology_eq:E_symmetric__torsion}
\\
&\bar{H}^{\mu\nu} 
=-D^{\left(\mu\right.}S^{\left.\nu\right)}\,,
\label{Cosmology_eq:Hbar_torsion_relation}\\
&H_{\alpha\beta}
=\bar{H}_{\alpha\beta}-h_{\alpha\beta}
\bar{H}_{\sigma}{}^{\sigma}\,,
\label{Cosmology_eq:Hbar_H_relation}\\
&u^{\gamma}\nabla_{\gamma}\bar{H}^{\alpha\beta}+
\frac{4}{3}\theta\bar{H}^{\alpha\beta}
=-2S_{\gamma}\varepsilon^{\mu\gamma\left(\alpha\right.}
\bar{H}^{\left.\beta\right)}{}_{\mu}\,,
\label{Cosmology_eq:dot_Hbar}\\
&\varepsilon_{\alpha}{}^{\mu\nu}D_{\nu}\bar{H}_{\mu\beta}= 
\frac{1}{3}\varepsilon_{\alpha\beta\gamma}S^{\gamma}\left\{
8\pi\mu+\Lambda-\frac{1}{3}\theta^{2}-
S^{\sigma}S_{\sigma}\right\}
\,;\label{Cosmology_eq:constraint_Hbar}
\end{align}

for the torsion field
\begin{align}
&\dot{S}^{\gamma}+\theta S^{\gamma}
=0\,,\label{Cosmology_eq:torsion_conservation}\\
&S^{\gamma}
=\Delta^{\gamma}\,
\label{Cosmology_eq:torsion_spin_density_relation}
\\
&\varepsilon_{\alpha}{}^{\mu\nu}D_{\mu}S_{\nu} 
=0\,,
\label{Cosmology_eq:curl_torsion}
\\
&S^{\mu}D_{\mu}S_{\nu} 
=0\,,\label{Cosmology_eq:torsion_parallel_transport}
\\
&D_{\nu}\left({S}^{\sigma}{S}_{\sigma}\right) 
=0\,;\label{Cosmology_eq:gradient_torsion_norm}
\end{align}

and for the matter variables
\begin{align}
&\dot{\mu}+\theta\left(\mu+p\right)  =0\,,
\label{Cosmology_eq:energy_density_conservation}
\\
&D_{\alpha}\mu  =0\,,\label{Cosmology_eq:Gradient_energy_density}
\\
&D_{\alpha}p  =0\,,\label{Cosmology_eq:Gradient_pressure}
\\
&q_{1}^{\alpha}  =0\,.
\label{Cosmology_eq:q1}
\end{align}

To close the system we have to either impose a function to model the
pressure, $p=p(x^\alpha)$, where $(x^\alpha)$ is some local
coordinate system on the manifold,
or relate $p$ with the energy density $\mu$ through
a barotropic equation of state, $p=p(\mu)$, i.e., 
\begin{align}
p=p(x^\alpha)\quad {\rm or}\quad p=p(\mu)\,.
\label{eos1}
\end{align}
Moreover, we see that there is no divergence
equation for the torsion vector field, other than that is must be equal
to minus the trace of $\bar{H}$. This is expected, since the geometry
and the field equations of the theory alone cannot determine the relation
between the spin density vector $\Delta^{\alpha}$ and the thermodynamical
variables $\mu$ and $p$: this is something that has to be provided
by a physical model for the matter. Therefore, to completely close
the system, we must either provide an ad hoc expression for
the spin density vector field, such that
$\Delta^{\alpha}=\Delta^{\alpha}\left( x^\alpha \right)$ or,
more physically motivated, an equation
that relates the spin density vector field with $\mu$ and $p$,
$\Delta^{\alpha}=\Delta^{\alpha}(\mu,p)$, i.e.,
\begin{align}
\Delta=\Delta(x^\alpha)\quad {\rm or}\quad\Delta^{\alpha}=
\Delta^{\alpha}(\mu,p)\,.
\label{eos2}
\end{align}

To further compare our results with those in the literature, notice that
from Eqs.~(\ref{Cosmology_eq:Hbar_torsion_relation}),
(\ref{Cosmology_eq:torsion_parallel_transport})
and (\ref{Cosmology_eq:gradient_torsion_norm})
we find that $\bar{H}_{\alpha\rho}S^{\rho}=0$.
This constraint, $\bar{H}_{\alpha\rho}S^{\rho}=0$, 
coincides with 
a constraint given in \citep{Hehl_1971,Puetzfeld_Obukhov_2007},
in which it is assumed 
that each element of volume of the fluid follows auto-parallel
curves and its rest mass is constant.
In our derivation of this constraint,
we have not assumed 
that the rest mass is constant, however
we have imposed that the fluid's volume
elements have zero acceleration
and it can be shown that this
implies that their rest mass is constant,
making the whole procedure consistent.

The previous set of equations can be written
in a somewhat more compact form. Remembering that the magnetic parts
of the Weyl tensor, $H$ and $\bar{H}$, are symmetric tensors,
Eqs.~(\ref{Cosmology_eq:Hbar_torsion_relation})
and (\ref{Cosmology_eq:curl_torsion}) can be replaced by the single
equation $\bar{H}^{\mu\nu}=-D^{\mu}S^{\nu}$. In that case the relation
$\bar{H}_{\alpha\rho}S^{\rho}=0$, can replace the
equations~(\ref{Cosmology_eq:torsion_parallel_transport})
and (\ref{Cosmology_eq:gradient_torsion_norm}). Notwithstanding,
we choose to keep all these properties explicit to avoid any confusion.

\subsection{Geometry of the 3-spaces for the universe with
homogeneous spinning fluid
\label{Subsec:Geometry_3_surf}}

We will now study some implications of the field equations,
Eqs.~\eqref{Cosmology_eq:Raychaudhuri_equation}--\eqref{Cosmology_eq:q1}. The
matter equations of state, Eqs.~\eqref{eos1} and \eqref{eos2}, will
not be used at this stage. We analyze in detail and obtain concrete
results related to the geometry of the 3-spaces, i.e.,
3-hypersurfaces, orthogonal to the congruence.

Without loss of generality, we will consider that the separation
vector field $n$, introduced in
Section~\ref{Subsec:The-separation-vector}, at each point is
orthogonal to the tangent vector field $u$, such that
$n^{\mu}u_{\mu}=0$ and $n^{\alpha}$ is spacelike, i.e., we will
consider the separation between points in the same orthogonal
hypersurface. This is always possible since, at a given point, we may
decompose a general separation vector in its components along $u$ and
orthogonal to it. Then, in the light of
Eq.~(\ref{Separation_eq:Dervivative_naUa}), in the considered setup,
the orthogonal part will stay orthogonal to $u$ as we move along its
integral curves. From Eqs.~\eqref{Separation_eq:derivative_n_relation}
and \eqref{Separation_eq:B_orthl_esv_general} and setting
$\sigma_{\alpha\beta}=0$ in accord to Eq.~\eqref{sigmaab=0} we have
\begin{equation}
\frac{D\left(n_{\alpha}n^{\alpha}\right)}{D\tau}=\frac{2}{3}
\theta n^{\alpha}n_{\alpha}\,,
\label{Cosmology_eq:1st_derivative_norm_separation} 
\end{equation}
where we used the notation $\frac{D}{D\tau}\equiv
u^{\gamma}\nabla_{\gamma}$, with $\tau$ being an affine parameter
parameterizing the integral curves of $u$, that is $\tau$ represents
up to a constant the proper time measured by an observer comoving with
the fiducial curve of the congruence. Taking another derivative along
$u$ we find
\begin{equation} 
\frac{D^{2}\left(n_{\alpha}n^{\alpha}\right)}{D\tau^{2}}=
\frac{2}{3}\dot{\theta}n^{\alpha}n_{\alpha}+\frac{4}{9}
\theta^{2}n^{\alpha}n_{\alpha}\,,
\label{Cosmology_eq:accel_norm_expansion}
\end{equation} relating the second derivative of
the square of the norm of the separation vector with the expansion 
coefficient of the congruence and its derivative.

Since $n$ is spacelike, we can define a length $\ell$ through the equation
\begin{equation} 
\ell=\sqrt{n^{\alpha}n_{\alpha}} \,,
\label{Cosmology_eq:lengthelldefinition}
\end{equation} 
where in general $\ell:\mathcal{M}\to \mathbb{R}$, that is, in some
local coordinate system, $\ell=\ell(x^\alpha)$, specifically of proper
time $\tau$ and the spatial coordinates on the
hypersurface. Nonetheless, since we have imposed $D_{\alpha}\theta=0$,
Eq.~\eqref{Dalpha=0}, it is always possible to define $n$ to represent
the separation vector between points at a fixed proper length at some
particular hypersurface, then
Eqs.~\eqref{Cosmology_eq:1st_derivative_norm_separation} and
\eqref{Cosmology_eq:lengthelldefinition} imply that
\begin{equation}
\ell=\ell\left(\tau\right)\,,
\label{elltau}
\end{equation}
and Eq.~(\ref{Cosmology_eq:1st_derivative_norm_separation}) can be written as 
\begin{equation}
\frac{1}{3}\theta=\frac{\dot{\ell}}{\ell}\,. 
\label{Cosmology_eq:scale_factor_expansion_definition} 
\end{equation}

Now, let $h_{ab}$, $ ^{3}R_{ab}$ and $^{3}R$ to represent,
respectively, the induced metric, the intrinsic Ricci tensor and the
intrinsic Ricci scalar of an orthogonal 3-hypersurface. Then, in the
considered setup, the Gauss embedding equation of differential
geometry yields the following relations between $ ^{3}R_{ab}$ and
$^{3}R$, and the induced metric, the kinematical and matter variables,
\begin{align}
& ^{3}R_{ab} =\frac{2}{3}h_{ab}\left(-\frac{1}{3}\theta^{2}-
S_{\sigma}S^{\sigma} + 8\pi\mu+\Lambda\right)\,, 
\label{Cosmology_eq:induced_Ricci_tensor_general}
\\
& ^{3}R =2\left( - \frac{1}{3}\theta^{2}-{S}^{\sigma}{S}_{\sigma} +
8\pi\mu+\Lambda \right)\,. 
\label{Cosmology_eq:Generalized_Friedmann} 
\end{align}
Equation~\eqref{Cosmology_eq:Generalized_Friedmann} is the generalized
Friedman equation for the Einstein-Cartan system we are interested. We
remark that for the type of torsion that is being considered,
Eq.~\eqref{EC_eq:simplifiedtorsion}, one can show that the induced
connection on the orthogonal slices to $u$ is the Levi-Civita
connection, hence $^{3}R_{ab}$ and $^{3}R$ represent the Ricci tensor
and Ricci scalar associated with the induced metric, $h_{ab}$.

From Eq.~(\ref{Cosmology_eq:Generalized_Friedmann}) we find
$
^{3}R\ell^{2}=-6\dot{\ell}^{2}-2{S}^{\sigma}{S}_{\sigma}\ell^{2}+
16\pi\mu\ell^{2}+2\Lambda\ell^{2}
$. Taking the derivative along $u$ of this equation and using the
Raychaudhuri equation~(\ref{Cosmology_eq:Raychaudhuri_equation}) and
the conservation equations~\eqref{Cosmology_eq:torsion_conservation}
and \eqref{Cosmology_eq:energy_density_conservation} yields
$\frac{D}{d\tau}\left(^{3}R\ell^{2}\right)=0$, that is, the quantity
$^{3}R\ell^{2}$ is a constant function between distinct hypersurfaces.
Indeed, using Eqs.~\eqref{Dalpha=0},
\eqref{Cosmology_eq:gradient_torsion_norm} and
\eqref{Cosmology_eq:Gradient_energy_density} we conclude that
\begin{equation}
^{3}R=\frac{6K}{\ell^2}\,,
\label{Cosmology_eq:R3_constant_general}
\end{equation}
where $K$ is some constant to be dealt with and the number 6
appears for convenience.
So, using Eqs.~\eqref{elltau} and \eqref{Cosmology_eq:R3_constant_general}
we have that the orthogonal 3-hypersurfaces are manifolds of
constant Ricci curvature, i.e.,
$\left.^{3}R\right|_{\tau}=\text{constant}$.
This result in
conjunction with
Eqs.~\eqref{Cosmology_eq:induced_Ricci_tensor_general} and
\eqref{Cosmology_eq:Generalized_Friedmann} leads us
to conclude that the Ricci
tensor of the 3-hypersurfaces is of the form $^{3}R_{ab}
=\frac{2K}{\ell^2} h_{ab}$,
i.e., a constant times the metric,
so that in the considered setup the
3-hypersurfaces are Einstein manifolds.
Now, in 3 dimensions the Riemann tensor is fully
characterized} by the Ricci tensor, specifically,
$^{3}R_{abcd}=2\left(^{3}R_{a\left[c\right.}
h_{\left.d\right]b}-\,{}^{3}R_{b\left[c\right.}
h_{\left.d\right]a}\right)-\,^{3}R\,h_{a\left[c\right.}
h_{\left.d\right]b}$,
which in the considered setup implies $^{3}
R_{abcd}=\frac{K}{\ell^{2}}\left(h_{ac}h_{db}-h_{ad}h_{cb}\right)$,
and
so the 
3-hypersurfaces are surfaces of
constant spatial curvature.
Since we assume that the
3-hypersurfaces are complete and simply-connected, we have that
the 3-hypersurfaces are isometric to the 3-hyperbolic space, to the
3-Euclidean space, or to the 3-sphere, in other words, in the
considered setup, the 3-hypersurfaces are isotropic and homogeneous
and the metric of the spacetime is a FLRW solution. Nonetheless, note
that due to the presence of the torsion tensor, the whole spacetime is
not described solely by the metric tensor.
In the light of these results, we can relate the value of the
integration constant $K$ in Eq.~\eqref{Cosmology_eq:R3_constant_general}
with the value of the
constant curvature of each 3-hypersurface,
i.e., $K=\left\{-1,0,1\right\}$, corresponding to the cases when the
orthogonal hypersurfaces are, for the natural topology, open and
hyperbolic, open and flat, or closed and spherical, respectively.
Note, however,
that depending on the topology, the solutions with $K=-1$ or $K=0$
need not be necessarily open, whereas the family of solutions with
$K=1$, to which the spherical solution belongs to, is necessarily
closed, see, e.g., \citep{Ellis_Maartens_MacCallum_Book}.

Although the metric tensor is a FLRW solution, the presence of torsion
modifies the geometry of the spacetime, in particular, we have found
that the Weyl tensor does not have to vanish,
see, e.g., Eq.~\eqref{Cosmology_eq:E_antisymmetric__torsion}.
This, of course, has
profound implications in the geometry of the spacetime and the type
of solutions that are allowed. In the light of the field equations,
we find the following results.
\begin{thm}
\label{Theorem:Negative_zero_curvature}
In the considered setup, if
$S_{\alpha}S^{\alpha}\neq0$ and
$D_{\alpha}\Delta^{\alpha}=f\left(\mu,p\right)$, where $f$ is an
arbitrary differentiable function, then
\begin{equation}
^{3}R\leq0\,,
\label{3Rleq0}
\end{equation}
i.e.,  $K=-1$ or $K=0$ in the FLRW metric.
Moreover,
each orthogonal hypersurface is flat, that is 
$^{3}R=0$, if and only if $\bar{H}^{\alpha\beta}=0$ for all points on
the hypersurface.
\end{thm}
\begin{proof}
From Eqs.~(\ref{Cosmology_eq:Hbar_torsion_relation})
and (\ref{Cosmology_eq:constraint_Hbar})
and $\bar{H}_{\alpha\rho}S^{\rho}=0$, we find the following relation
\begin{equation}
S^{\delta}D_{\delta}\left(D_{\alpha}S^{\alpha}\right)=-
\frac{2}{3}S_{\delta}S^{\delta}\left\{ 8\pi\mu+\Lambda-
\frac{1}{3}\theta^{2}-S^{\sigma}S_{\sigma}\right\} -
\bar{H}^{\mu\nu}\bar{H}_{\mu\nu}\,.
\label{Cosmology_eq:intermediate_proof_1}
\end{equation}
Note
that for the type of torsion that we
are considering, Eq.~\eqref{EC_eq:simplifiedtorsion},
one has that $D_{\alpha}Y^{\alpha}$
is indeed the divergence of a vector field $Y$ orthogonal to $u$,
so that $D_{\alpha}S^{\alpha}$ is the divergence
of the torsion vector $S$.
Imposing that the divergence of the spin density vector is a differentiable
function of the energy density $\mu$ and the pressure $p$, that
is $D_{\alpha}\Delta^{\alpha}=f\left(\mu,p\right)$, and
using Eqs.~(\ref{Cosmology_eq:torsion_spin_density_relation}),
(\ref{Cosmology_eq:Gradient_energy_density}) and
(\ref{Cosmology_eq:Gradient_pressure})
implies that $D_{\delta}\left(D_{\alpha}S^{\alpha}\right)=0$.
Using this result
and the Friedman equation~(\ref{Cosmology_eq:Generalized_Friedmann}),
Eq.~(\ref{Cosmology_eq:intermediate_proof_1}) yields
\begin{equation}
\bar{H}^{\mu\nu}\bar{H}_{\mu\nu}=-\frac{^{3}R}{3}S_{\delta}S^{\delta}.
\label{Cosmology_eq:intermediate_proof_2}
\end{equation}
Since the hypersurfaces orthogonal to the tangent vector field $u$
are Riemannian
manifolds, the terms $\bar{H}^{\mu\nu}\bar{H}_{\mu\nu}$
and $S_{\delta}S^{\delta}$
must be non-negative, therefore a consistent solution of the field
equations with $S_{\delta}S^{\delta}\neq0$ must verify $^{3}R\leq0$,
i.e.,  $K=-1$ or $K=0$ in the FLRW metric.
Using this same argument, it follows that if $^{3}R=0$,
then $\bar{H}^{\mu\nu}=0$.
Of course, trivially, if $\bar{H}^{\mu\nu}=0$ and $S_{\delta}S^{\delta}\neq0$,
then $^{3}R=0$.
\end{proof}

The result in Theorem~\ref{Theorem:Negative_zero_curvature} is
quite surprising. In the considered setup and for a nonvanishing
torsion vector field, the orthogonal hypersurfaces
that foliate the spacetime must either
have negative curvature or be
Ricci flat. In addition, we find the following result:
\begin{thm}
\label{Theorem:Global_not_closed}In the considered setup,
if the torsion $S$ is such that
$S_{\alpha}S^{\alpha}\neq0$
for all points on the hypersurfaces orthogonal to the congruence associated
with $u$, then the hypersurfaces cannot be closed.
\end{thm}
\begin{proof}
Let us start by recalling that we
have imposed the congruence $u$ to be
hypersurface orthogonal. We then chose a frame where the orthogonal
slices $ \mathcal{V} $ are hypersurfaces, hence there
exists an embedding between each hypersurface $\mathcal{V}$ and a
Riemannian manifold $\left(\mathcal{V},h\right)$, where $h$ represents
the induced metric and, for the type of torsion that we
are considering
in this section, the induced torsion tensor is zero. Since an embedding
exists, we can pull-back and push-forward nonvanishing orthogonal
tensor fields in $\left(\mathcal{M},g,S\right)$ to nonvanishing
tensor fields in $\left(\mathcal{V},h\right)$. In particular, the
pull-back of the projected covariant derivatives of an orthogonal 1-form field
${Y}_{\alpha}$, that is ${Y}_{\alpha}u^{\alpha}=0$, is given by
$D_{a}{Y}_{b}$, where ${Y}_{a}$ represents the pull-back of ${Y}_{\alpha}$
and defines the induced connection in $\left(\mathcal{V},h\right)$,
which is simply the Levi-Civita connection associated with $h$.
Of course, ${Y}_{\alpha}$ represents the components
of ${Y}\in T_{p}^{*}\mathcal{M}$
in a local coordinate system and ${Y}_{a}$ the components of
${Y}\in T_{p}^{*}\mathcal{V}$
in a local coordinate system,
where $T_{p}^{*}$ means the cotangent space of the corresponding manifold at
the point $p$, however, although an abuse of language,
it is much simpler and became kind of a convention to distinguish
between the two tensors by using Greek and Latin letters.

From Eq.~(\ref{Cosmology_eq:curl_torsion}), if we define ${S}_{a}$
as the pull-back of the 1-form $S_{\alpha}$, we find that it verifies
\begin{equation}
\varepsilon_{a}{}^{bc}D_{b}{S}_{c}=0\Leftrightarrow
\varepsilon_{a}{}^{bc}\partial_{b}{S}_{c}=0\,,
\end{equation}
where $\varepsilon_{abc}$ represents the
Levi-Civita tensor in $\left(\mathcal{V},h\right)$.
Therefore, ${S}_{a}$ is an exact 1-form, that is, there exists
a function $\phi$, such that ${S}=d\phi$. Moreover, since
$h$ is a Riemannian metric, it is nondegenerate, hence the condition
${S}_{a}{S}^{a}\neq0$ implies that $d\phi\neq0$. To
clarify, the induced torsion tensor on $\left(\mathcal{V},h\right)$
is zero, meaning that the manifold is endowed with only the Levi-Civita
connection. However, ${S}_{a}$ does not have to be zero, and
it should be regarded simply as a 1-form field in $T^{*}\mathcal{V}$
with no relation with the connection.

Now, if $\mathcal{V}$ is closed, it is, by definition,
compact and it has no boundary, then, from Stokes
Theorem, we have $\int_{\mathcal{V}}d\phi=0$. However, $d\phi\neq0$,
hence $d\phi$ is a volume form and its integral over $\mathcal{V}$
cannot be zero.
\end{proof}

This result is also surprising. In addition to the
result of Theorem~\ref{Theorem:Negative_zero_curvature} asserting
that the orthogonal
hypersurfaces cannot have positive curvature, i.e., $K\leq0$,
we see that
Theorem~\ref{Theorem:Global_not_closed} establishes that they
also cannot be closed, limiting the topology of these solutions.  This
is indeed a great disparity between the theory of Einstein-Cartan and
general relativity, since in the latter there is no limitation in the
sign of $K$ nor on the topology.

\vskip 1cm

The intermediate results for the proof of
Theorem~\ref{Theorem:Negative_zero_curvature}
also allow us to infer the behavior of the magnetic part of the Weyl
tensor. Considering
Eqs.~\eqref{Cosmology_eq:torsion_conservation},
\eqref{Cosmology_eq:scale_factor_expansion_definition},
\eqref{Cosmology_eq:R3_constant_general} and
\eqref{Cosmology_eq:intermediate_proof_2}
we have:

\newpage

\begin{prop}
\label{Proposition:Behavior_H_square}In the considered setup, if
$S_{\alpha}S^{\alpha}\neq0$, $^{3}R<0$ and $D_{\alpha}
\Delta^{\alpha}=f\left(\mu,p\right)$,
where $f$ is an arbitrary differentiable function,
then $\bar{H}^{\mu\nu}\bar{H}_{\mu\nu}\sim\frac{1}{\ell^{8}}$.
\end{prop}

Proposition~\ref{Proposition:Behavior_H_square} establishes the
behavior of the tensor $\bar{H}$, and similarly for $H$, in terms
of the scale factor $\ell$, so that if the spacetime is expanding,
$\bar{H}^{\mu\nu}\bar{H}_{\mu\nu}$ tends to zero as $\frac{1}{\ell^{8}}$.
In the next section we  further study the tensor $\bar{H}$,
in particular we  show that it is possible to derive a wave equation
for $\bar{H}$ and then study its solutions.

We use the results obtained here to clarify some confusion
regarding the possibility to consider a torsion
caused by an intrinsic spin of matter
in a cosmological context.
In \citep{Tsamparlis_1981} it
was shown that, under certain conditions, the symmetries of the metric
tensor, in the form of Killing vector fields, are also symmetries of
the torsion tensor. Under those conditions, of course, it was then
found that a torsion tensor having its origin in the intrinsic spin of matter, 
a torsion tensor of the form $S_{\alpha\beta}{}^{\gamma}=
\varepsilon_{\alpha\beta\mu}S^{\mu}u^{\gamma}$, where $S^{\mu}$ is a
spacelike vector field, is not compatible with the cosmological
principle. Since the publication of \citep{Tsamparlis_1981}, much of
the literature considering an isotropic and homogeneous universe in
the Einstein-Cartan theory has completely disregarded a torsion tensor
of the previous form. However, we have to analyze the conditions under
which it is valid
the assertion that symmetries of the metric tensor are also
symmetries of the torsion tensor. The pivotal condition is
that the symmetries of the metric are also symmetries of the metric
stress-energy tensor, however, in \citep{Tsamparlis_1981} it is
clearly stated that this very strong condition is imposed ad hoc and,
contrary to the theory of general relativity, does not follow from the field
equations of the Einstein-Cartan theory. Nonetheless, it is defended
that this is a reasonable assumption if the Einstein-Cartan theory is
considered, in some sense, as a slight modification to general
relativity. This, however, in general is not the case. As we can
readily infer from the structure equations
\eqref{Cosmology_eq:Raychaudhuri_equation}--\eqref{Cosmology_eq:q1},
the Einstein-Cartan theory, in general, is not a slight modification
to general relativity. For instance, notice that the torsion tensor
directly couples and acts as a source to the Weyl tensor.
Of course,
models with a vanishing Weyl tensor, as it is the case in general
relativity, or a nonvanishing Weyl tensor, as
generically presented here for the Einstein-Cartan theory
we have been considering, 
represent very
distinct physical setups.
As shown above, the
torsion tensor does not have to have the same symmetries of the metric
tensor and the model just constructed is a consistent solution of
the Einstein-Cartan theory in a cosmological context for a universe
permeated by an isotropic and homogeneous matter fluid.

\newpage

\section{Gravitational waves in relativitstic cosmology in
Einstein-Cartan theory}
\label{Subsec:Wave-equation}

\subsection{Derivation of the gravitational wave
equation
for the isotropic universe}

Comparing
Eqs.~\eqref{Cosmology_eq:Raychaudhuri_equation}--\eqref{Cosmology_eq:q1}
with those found in the theory of general relativity for a homogeneous
and isotropic spacetime, see, e.g.,
\citep{Ellis_vanElst_1999,Ellis_Maartens_MacCallum_Book}, we see that
a glaring difference is that the Weyl tensor is, in general, not
identically zero. This contrast between the two theories has profound
implications in the evolution of the spacetime geometry and of the
matter fluid. Indeed, in the previous section we have found that a
nonvanishing Weyl tensor restricts the allowed geometry and topology of the
orthogonal hypersurfaces, a restriction that does not exist in general
relativity. In this section, we study further the Weyl tensor and its
effects on the evolution of the spacetime curvature.

The Weyl tensor is known to be related with gravitational waves and
tidal forces, which in fact are interconnected
phenomena.  In the
model we are considering here, we
have found that torsion and its derivatives act as a source for the
Weyl tensor components, hence a natural step to understand the
solutions of the structure equations is to study the presence of
gravitational waves
induced by the matter intrinsic spin. Due to the presence of
torsion, if $K=-1$, the magnetic part of the Weyl tensor is
nonvanishing. In this subsection, we will show that if the torsion
tensor is caused by the matter fluid, the traceless part of $\bar{H}$,
obeys a wave equation.
These wave equations can be formally solved, explicitly
showing that in a nonstatic universe the presence of intrinsic spin leads
to the generation and emission of gravitational waves.
From Eq.~\eqref{Cosmology_eq:Hbar_H_relation}, $H$ and $\bar{H}$ have
the same traceless part, but distinct trace, namely $H_{\left\langle
\alpha\beta\right\rangle }=\bar{H}_{\left\langle
\alpha\beta\right\rangle }$ and
$H_{\alpha}{}^{\alpha}=-2\bar{H}_{\alpha}{}^{\alpha}$. In fact, it is
straightforward to show that $H$ and $\bar{H}$ have the same
eigenvectors, but associated with distinct eigenvalues. Then, in this
section we will focus on studying $\bar{H}$ and all results are
directly extended to $H$.

From the Ricci
identity~(\ref{Conventions_eq:Riemann_tensor_definition}) and the
field equations, we find, in the considered setup, the
following expression for the projected derivative of the divergence of
$\bar{H}$,
\begin{equation}
\begin{aligned}
\left(D_{\mu}D_{\alpha}\bar{H}^{\mu}{}_{\beta}\right)-
\left(D_{\alpha}D_{\mu}\bar{H}^{\mu}{}_{\beta}\right)=
& \left(8\pi\mu+\Lambda-S_{\sigma}S^{\sigma}-
\frac{1}{3}\theta^{2}\right)\bar{H}_{\left\langle
 \alpha\beta\right\rangle }\\ &
 +\frac{1}{3}\theta
 \left(\bar{H}^{\mu}{}_{\mu}\varepsilon_{\alpha\beta}{}^{\gamma}-
 \bar{H}_{\mu\beta}\varepsilon_{\alpha}{}^{\mu\gamma}-
 \bar{H}^{\mu}{}_{\alpha}\varepsilon_{\mu\beta}{}^{\gamma}\right)
 S_{\gamma}\,.
\end{aligned}
\end{equation}
On the other hand, Eq.~(\ref{Cosmology_eq:constraint_Hbar}) implies
\begin{equation}
D_{\mu}\left(D^{\alpha}\bar{H}^{\mu\beta}\right)-
D_{\mu}\left(D^{\mu}\bar{H}^{\alpha\beta}\right)+
\frac{1}{3}\left(\bar{H}^{\alpha\beta}-
h^{\alpha\beta}\bar{H}^{\mu}{}_{\mu}\right)
\left(8\pi\mu+\Lambda-\frac{1}{3}\theta^{2}-
S^{\sigma}S_{\sigma}\right)=0\,.
\end{equation}
Taking the derivative of Eq.~(\ref{Cosmology_eq:dot_Hbar})
we find
\begin{equation}
\frac{D^{2}}{d\tau^{2}}\bar{H}^{\alpha\beta}+\frac{4}{3}
\bar{H}^{\alpha\beta}\left(\dot{\theta}-\frac{4}{3}
\theta^{2}+3S_{\gamma}S^{\gamma}\right)-\frac{22}{3}
\theta S_{\gamma}\varepsilon^{\mu\gamma\left(\alpha\right.}
\bar{H}^{\left.\beta\right)}{}_{\mu}=
2\bar{H}^{\mu}{}_{\mu}\left(h^{\alpha\beta}S^{\delta}
S_{\delta}-S^{\alpha}S^{\beta}\right)\,.
\label{neweq}
\end{equation}
Gathering these results, yields
\begin{equation}
\begin{aligned}
\left(\frac{D^{2}}{d\tau^{2}}-
D_{\mu}D^{\mu}\right)\bar{H}_{\alpha\beta}+D_{\alpha}
\left(D_{\beta}\bar{H}_{\mu}{}^{\mu}\right)+
2\bar{H}_{\left\langle \alpha\beta\right\rangle }
\left(8\pi\mu+\Lambda-\frac{1}{3}\theta^{2}+
S^{\sigma}S_{\sigma}\right)+\frac{4}{3}
\bar{H}_{\alpha\beta}\left(\dot{\theta}-
\frac{4}{3}\theta^{2}\right)\\
+\frac{1}{3}\theta\left(\bar{H}^{\mu}{}_{\mu}
\varepsilon_{\alpha\beta}{}^{\nu}-
12\varepsilon^{\mu\nu}{}_{\alpha}\bar{H}_{\mu\beta}-
10\varepsilon^{\mu\nu}{}_{\beta}
\bar{H}_{\alpha\mu}\right)S_{\nu}+
2\bar{H}^{\mu}{}_{\mu}
S_{\left\langle \alpha\right.}
S_{\left.\beta\right\rangle } & =0\,.
\end{aligned}
\label{Cosmology_eq:almost_wave_equation}
\end{equation}
Note that the operator $\frac{D^{2}}{d\tau^{2}}-D_{\mu}D^{\mu}$
is not the wave operator, since it is defined in terms of the total
connection, nonetheless, it is equal to the wave operator plus
terms in $\bar{H}$ and its first derivatives.

Now, the term $D_{\alpha}
\left(D_{\beta}\bar{H}_{\mu}{}^{\mu}\right)$
in the left-hand side of the previous equation does not have to be
zero. Since this term is a
second order derivative of $\bar{H}$,
in general the components of $\bar{H}$ are not solutions of a wave
equation. Notwithstanding, it is physically reasonable to consider
that the divergence of the spin density vector is a differentiable
function of the energy density $\mu$ and the pressure $p$, that
is $D_{\alpha}\Delta^{\alpha}=f\left(\mu,p\right)$: this expresses
the idea that $\Delta^{\alpha}$ has the matter fields as its source.
In that case, Eqs.~(\ref{Cosmology_eq:torsion_spin_density_relation}),
(\ref{Cosmology_eq:Gradient_energy_density}) and
(\ref{Cosmology_eq:Gradient_pressure})
imply that $D_{\beta}\bar{H}_{\mu}{}^{\mu}=0$. Therefore, we
have the following result:
\begin{prop}
\label{Proposition:Wave_equation_Hbar}In the considered setup, if
$D_{\alpha}\Delta^{\alpha}=f\left(\mu,p\right)$, where $f$ is an
arbitrary differentiable function, the magnetic part of the Weyl tensor
$\bar{H}$ verifies the following wave equation for the symmetric part
without trace
\begin{align}
\tilde{\square}\bar{H}_{\left\langle \alpha\beta\right\rangle
}+2\bar{H}_{\left\langle \alpha\beta\right\rangle
}\left(8\pi\mu+\Lambda-\frac{1}{3}\theta^{2}-S^{\sigma}
S_{\sigma}\right)+\frac{4}{3}\bar{H}_{\left\langle
\alpha\beta\right\rangle
}\left(\dot{\theta}-\frac{4}{3}\theta^{2}\right) =0
\label{Cosmology_eq:Wave_equation}
\end{align}
where $\tilde{\square}\bar{H}_{\alpha\beta}:=
\left(\frac{\tilde{D}^{2}}{d\tau^{2}}-
\tilde{D}_{\mu}\tilde{D}^{\mu}\right)\bar{H}_{\alpha\beta}$,
defined in terms of the Levi-Civita connection, represents the wave
operator, 
and $\bar{H}$ verifies further
the following evolution equation for the
trace
\begin{align}
\frac{D}{d\tau}\bar{H}_{\alpha}{}^{\alpha}+\frac{4}{3}
\bar{H}_{\alpha}{}^{\alpha}\theta  =0\,.
\label{Cosmology_eq:Trace_Hbar_evolution}
\end{align}
\end{prop}
Thus, we have then found that
the traceless part of $\bar{H}$ verifies a homogeneous wave equation
and 
the trace of $\bar{H}$ verifies a first order
ODE. Before
we proceed to study the solutions of the previous set of
two equations,
we remark that the coefficient of the second term in the left-hand
side of Eq.~(\ref{Cosmology_eq:Wave_equation}) is simply the Ricci
scalar of the orthogonal hypersurfaces,
Eq.~(\ref{Cosmology_eq:Generalized_Friedmann}).

\subsection{The solutions}

Theorem~\ref{Theorem:Negative_zero_curvature} establishes that the
orthogonal hypersurfaces to $u$ cannot have positive curvature and
if these have zero Ricci curvature, $\bar{H}$ must be identically
zero. Therefore, the only nontrivial solutions of
Eqs.~(\ref{Cosmology_eq:Wave_equation})
and (\ref{Cosmology_eq:Trace_Hbar_evolution}) that are of physical
interest are those where $K=-1$. Notwithstanding, formally the treatment
below is largely independent of the sign of $K$ and we only have
to specify the allowed values of $K$ when we consider the initial
conditions. Therefore, in an effort to be pedagogical about the covariant
analysis of gravitational waves in a cosmological setting, we will
keep the discussion as general as possible and only when studying
the behavior of the solutions we will particularize to $K=-1$.

The second equation in
Proposition~\ref{Proposition:Wave_equation_Hbar},
Eq.~\eqref{Cosmology_eq:Trace_Hbar_evolution}, is a first-order ODE
and can be readily integrated in terms of the characteristic length
$\ell$. Using Eq.~\eqref{Cosmology_eq:scale_factor_expansion_definition}
we find $\bar{H}_{\alpha}{}^{\alpha}=\frac{C}{\ell^{4}}$, where
$C\in\mathbb{R}$. Note that to find
Eqs.~\eqref{Cosmology_eq:Wave_equation} and
\eqref{Cosmology_eq:Trace_Hbar_evolution} we have imposed that
$D_{\alpha}\Delta^{\alpha}=f\left(\mu,p\right)$, which implies
$D_{\alpha}\bar{H}_{\mu}{}^{\mu}=0$.  On the other hand, finding the
solutions of the wave equation
given in Eq.~(\ref{Cosmology_eq:Wave_equation}) is
more involved. In that regard, we will assume that the spatial and
proper-time, $\tau$, dependence of $\bar{H}_{\left\langle
\alpha\beta\right\rangle }$ are separable. Then, 
we will consider the eigenfunctions of the covariant Laplace-Beltrami
operator $\tilde{D}_{\mu}\tilde{D}^{\mu}$ and expand
$\bar{H}_{\left\langle \alpha\beta\right\rangle }$ over these
eigenfunctions, such that
\begin{equation}
\bar{H}_{\left\langle \alpha\beta\right\rangle
}=\sum_{k}\mathrm{h}_{k}^{\left(0\right)}
Q_{\alpha\beta}^{\left(0\right),k}+\mathrm{h}_{k}^{\left(1\right)}
Q_{\alpha\beta}^{\left(1\right),k}+\mathrm{h}_{k}^{\left(2\right)}
Q_{\alpha\beta}^{\left(2\right),k}\,,
\end{equation}
where we have used a compact notation to unify the two possibilities of
$k$ taking discrete or continuous values, such that the symbol
$\sum_{k}$ is to be understood as either a discrete sum, if the
hypersurfaces orthogonal to $u$ have positive curvature, $K=1$, or as
an integral over a continuously varying index, if these have zero or
negative curvature; also the coefficients
$\mathrm{h}_{k}^{\left(0\right)}$, $\mathrm{h}_{k}^{\left(1\right)}$
and $\mathrm{h}_{k}^{\left(2\right)}$ are in general functions of the
proper time $\tau$ and
$\dot{Q}_{\alpha\beta}^{\left(0\right),k}=
\dot{Q}_{\alpha\beta}^{\left(1\right),k}=
\dot{Q}_{\alpha\beta}^{\left(2\right),k}=0$.
Moreover, the minimum values of the eigenvalues $k^{2}$ are
$k^{2}=0,1,3$ if, respectively, the orthogonal hypersurfaces are,
for the natural topology, flat,
open or closed.  Nonetheless, bear in mind the
since $\tilde{D}_{\alpha}Q^{0}=0$, even if $k^{2}=0$ is an eigenvalue
of the Helmholtz equation, we have
$Q_{\alpha\beta}^{\left(0\right),0}=0$, see
Appendix~\ref{Section:Appendix}.
This type of decomposition is known as scalar-vector-tensor
decomposition due to some properties of the harmonics
$Q_{\alpha\beta}^{\left(0\right),k}$,
$Q_{\alpha\beta}^{\left(1\right),k}$ and
$Q_{\alpha\beta}^{\left(2\right),k}$, in particular we have that the
curl of $Q_{\alpha\beta}^{\left(0\right),k}$, defined as
$\varepsilon_{\left(\alpha\right|}{}^{\mu\nu}\tilde{D}_{\nu}
Q_{\mu\left|\beta\right)}^{\left(0\right),k}$,
is identically zero,
$\tilde{D}^{\beta}\tilde{D}^{\alpha}Q_{\alpha\beta}^{\left(1\right),k}=0$
and $\tilde{D}^{\alpha}Q_{\alpha\beta}^{\left(2\right),k}=0$. For
completeness, we list various properties of the scalar, vector and
tensor harmonics in the Appendix~\ref{Section:Appendix}.

From Eq.~(\ref{Cosmology_eq:constraint_Hbar}), we have that
$\text{curl}\,\bar{H}_{\alpha\beta}
\equiv\varepsilon_{\left(\alpha\right|}{}^{\mu\nu}D_{\nu}
\bar{H}_{\mu\left|\beta\right)}$
vanishes. Hence, $\bar{H}$ can be described solely by the scalar
harmonics $Q_{\alpha\beta}^{\left(0\right),k}$. Then, substituting the
expansion $\bar{H}_{\left\langle \alpha\beta\right\rangle
}=\sum_{k}\mathrm{h}_{k}^{\left(0\right)}Q_{\alpha\beta}^{\left(0\right),k}$
in the wave equation
given in Eq.~(\ref{Cosmology_eq:Wave_equation}), the harmonics
decouple and we find for each $k$ the equation
\begin{equation}
\ddot{\mathrm{h}}_{k}^{\left(0\right)}+
\mathrm{h}_{k}^{\left(0\right)}
\left[\frac{k^{2}}{\ell^{2}}+\frac{4}{3}\left(\dot{\theta}-
\frac{4}{3}\theta^{2}\right)\right]=0\,.\label{Wave_eq:h_ddot}
\end{equation}
Hence, the expansion coefficients verify an equation for an harmonic
oscillator with variable frequency, leading us to conclude that, in
an dynamic universe, the presence of intrinsic spin may induce the
emission of gravitational waves. Introducing the Hubble parameter
$\mathrm{H}\equiv\frac{1}{3}\theta$, the conformal time variable
$t$, defined such that $dt=\ell^{-1}d\tau$, and writing
$\mathrm{h}_{k}^{\left(0\right)}=\frac{f_{k}\left(t\right)}{\ell^{4}}$,
we find that these are the solutions of Eq.~(\ref{Wave_eq:h_ddot})
if each $f_{k}\left(t\right)$ verifies
\begin{equation}
\frac{d^{2}f_{k}}{dt^{2}}-
9\ell\mathrm{H}\frac{df_{k}}{dt}+k^{2}f_{k}=0\,.
\label{Wave_eq:h_ddot_f_ddot}
\end{equation}

Equation~(\ref{Wave_eq:h_ddot_f_ddot}) takes a surprisingly simple
form and all dependencies of the matter model are encapsulated in
the quantity $\ell\mathrm{H}$, defined as the inverse comoving Hubble
radius $R_{\mathrm{H}}$, i.e.,
$R_{\mathrm{H}}\equiv\left(\ell\mathrm{H}\right)^{-1}$.
Now, to integrate Eq.~(\ref{Wave_eq:h_ddot_f_ddot})
one has either to assume a model for the matter fluid
or to resort to
solutions valid within certain regimes.
We stick to the second alternative. 
For this, note that 
in the light of Proposition~\ref{Proposition:Behavior_H_square},
a consistent solution must be such that the functions $f_{k}$ are
bounded. Then, we can analyze the cases for which
$\frac{k^{2}}{\ell\mathrm{H}}\gg1$ and 
$\frac{k^{2}}{\ell\mathrm{H}}\ll1$.

In the regime where $\frac{k^{2}}{\ell\mathrm{H}}\gg1$
and such that the term in $\frac{df_{k}}{dt}$ is negligible,
with
no need for
specifying the matter fields that permeate the spacetime,
and further assuming $f_{k}$ and its derivatives up to second
order are bounded, the
solutions of Eq.~(\ref{Wave_eq:h_ddot}) for the higher order modes
are of the form
\begin{equation}
\mathrm{h}_{k}^{\left(0\right)}=\frac{c_{1}
\cos\left(kt\right)+c_{2}
\sin\left(kt\right)}{\ell^{4}}\,,
\label{Wave_eq:coefficient_late_time1}
\end{equation}
where the integration constants $c_{1}$ and $c_{2}$
might change for each $h_{k}^{\left(0\right)}$. This result makes
it clear that in the considered setup, a nonvanishing $\bar{H}$
characterizes gravitational waves induced by intrinsic spin. Moreover,
we see that in an expanding universe these waves are strongly damped,
as it was found in Proposition~\ref{Proposition:Behavior_H_square}. 

In the regime where
$\frac{k^{2}}{\ell\mathrm{H}}\ll1$, and also with
no need for
specifying the matter fields that permeate the spacetime,
in the light of Theorem~\ref{Theorem:Negative_zero_curvature}, the
only nontrivial solutions for Eq.~(\ref{Wave_eq:h_ddot}) that are
of physical interest in the considered model are those where $K=-1$.
In that case, the expansion coefficient $k$ takes continuous values
and $k\geq1$. Now, the comoving Hubble radius
verifies $\dot{R}_{\mathrm{H}}=-\ddot{\ell}R_{\mathrm{H}}^{2}$.
Then, in an accelerating expanding universe, $R_{\mathrm{H}}$ is
a decreasing function of the proper time. Therefore,  for $K=-1$,
the regime where
$\frac{k^{2}}{\ell\mathrm{H}}\ll1$ represents the late-time behavior
of the lower order modes of the spin induced gravitational waves in
an accelerating expanding universe. In this regime, assuming we can
neglect the term in $f_{k}$ in Eq.~(\ref{Wave_eq:h_ddot_f_ddot})
and disregarding runaway solutions, we find that
\begin{equation}
\mathrm{h}_{k}^{\left(0\right)}=\frac{\text{constant}}{\ell^{4}}\,,
\label{Wave_eq:coefficient_late_time2}
\end{equation}
where the integration constant might change for each $k$,
confirming once again that Proposition~\ref{Proposition:Wave_equation_Hbar}
is consistent with the results found in
subsection~\ref{Subsec:Geometry_3_surf}.

\section{Tidal effects and dynamics of the cosmic fluid in
relativistic cosmology in Einstein-Cartan theory}
\label{Subsec:Tidal-effects-and-dynamics}

\subsection{Tidal effects}
\label{Subsec:Tidal-effects}

In addition to the magnetic part of the Weyl tensor,
the electric part of the Weyl tensor is also not
identically zero in the presence
of torsion. Therefore, tidal effects, i.e., the relative
accelerations of nearby particles, suffer modifications
when compared to the
theory of general relativity.

The general formula for the tidal displacement in the presence of
torsion is
\begin{equation}
\frac{D^{2}n^{\delta}}{d\tau^{2}}=
n^{\mu}\nabla_{\mu}a^{\delta}+R_{\alpha\beta\gamma}{}^{\delta}
n^{\alpha}u^{\beta}u^{\gamma}+2u^{\sigma}\nabla_{\sigma}
\left(S_{\alpha\beta}{}^{\delta}u^{\alpha}n^{\beta}\right)\,,
\label{deviation1}
\end{equation}
where $n$ represents the separation vector introduced
in subsection~\ref{Subsec:The-separation-vector}.
In the setup we are considering, Eq.~\eqref{deviation1}
reduces to $\frac{D^{2}n^{\delta}}{d\tau^{2}}  =
R_{\sigma\mu\nu}{}^{\delta}n^{\sigma}u^{\mu}u^{\nu}$,
which is the familiar formula for geodesic deviation. Assuming without loss of
generality that the separation vector is initially orthogonal to the
tangent vector field $u$, i.e., $n^{\mu}u_{\mu}=0$,  this
can further be manipulated to have the form
\begin{equation}
\frac{D^{2}n^{\delta}}{d\tau^{2}}
=n^{\mu}\left(\frac{1}{2}R_{\mu}{}^{\delta}+\frac{1}{2}
R_{\mu\nu}u^{\nu}u^{\delta}-E_{\mu}{}^{\delta}\right)
-\frac{1}{2}\left(R_{\mu\nu}u^{\mu}u^{\nu}+\frac{1}{3}R\right)n^{\delta}\,.
\label{deviation2}
\end{equation}
We see that Eq.~\eqref{deviation2}
explicitly shows the influence of the electric part of the
Weyl tensor in the tidal displacement. Using
Eqs.~(\ref{Conventions_eq:Weyl_tensor_definition}),
(\ref{EC_eq:field_equations_Ricci}),
(\ref{Cosmology_eq:E_antisymmetric__torsion}) and
(\ref{Cosmology_eq:E_symmetric__torsion}), we find the following
expression for the relative acceleration between two infinitesimally
close test particles in the considered model
\begin{align}
\frac{D^{2}n^{\delta}}{d\tau^{^{2}}} &
=n^{\alpha}\left(\frac{1}{3}\varepsilon_{\alpha\gamma}{}^{\delta}
\theta\mathrm{S}^{\gamma}+\mathrm{S}_{\alpha}\mathrm{S}^{\delta}\right)
+\frac{1}{3}\Bigl(\Lambda-\mathrm{S}^{\sigma}\mathrm{S}_{\sigma}-
4\pi\left(\mu+3p\right)\Bigr)n^{\delta}\,.
\label{Cosmology_eq:geodesic_deviation_particular}
\end{align}
From Eq.~(\ref{Cosmology_eq:geodesic_deviation_particular}), we see
that the presence of intrinsic spin induced torsion causes a
distortion of the fluid as measured by an observer comoving with the
fluid that couples directly with torsion. To interpret this result, it
is clearer to consider Eq.~\eqref{omegagammaSgamma}, that is, the
presence of torsion induces a rotation of the frame of the
observer. This is a well known effect of the torsion tensor, which in
fact lead to its name: a test particle, or an element of volume of the
fluid, that couples directly with the torsion field, in general, will
have its frame rotated.  Then, the distortion of the fluid described
by Eq.~(\ref{Cosmology_eq:geodesic_deviation_particular}) is caused by
the relative acceleration of the rotation of the fiducial observer's
frame due to the presence of torsion.  This rotation of the frames is
due to the spacetime geometry, however it is not intrinsic to the
motion of the fluid, i.e., the fluid is irrotational since elements of
volume of the fluid follow metric geodesics of the spacetime, whose
metric is described by a FLRW solution. Thus, observers, that do not
couple directly with the torsion tensor, will not measure any relative
rotation between different points in the fluid.  To see this, consider
an observer that does not couple directly with the torsion tensor and
only perceives the effects of the intrinsic spin of matter through the
metric tensor, such that its world line is a metric geodesic of the
spacetime and its 4-velocity coincides with $u$, the tangent vector of
the congruence.  For this type of observer, the geodesic deviation
equation, in the considered setup, reads
$\frac{\tilde{D}^{2}n^{\delta}}{d\tau^{2}}=\frac{1}{3}
\Bigl(\Lambda+2\mathrm{S}^{\sigma}\mathrm{S}_{\sigma}-
4\pi\left(\mu+3p\right)\Bigr)n^{\delta}$, where
$\frac{\tilde{D}^{2}n^{\delta}}{d\tau^{2}}\equiv
u^{\mu}\tilde{\nabla}_{\mu}\left(u^{\nu}
\tilde{\nabla}_{\nu}n^{\delta}\right)$ and $\tilde{\nabla}$ represents
the Levi-Civita connection. We see, then, that such observer does not
measure any relative change in the rotation between nearby elements of
volume of the fluid. This type of observer will only measure a
relative acceleration of the distance between infinitesimally close
test particles. These results exactly express the discussion in
\citep{Puetzfeld_Obukhov_2007} where it was determined in the context
of any metric-affine gravity theory, particles with no intrinsic
hypermomentum will not directly experience the effects of torsion.

In addition, relative acceleration of the squared distance between
infinitesimally close test particles is a physical observable, hence
both type of observers, namely,
those that couple directly to torsion
and those that do not, 
will agree on its magnitude. Using the Raychaudhuri
equation~(\ref{Cosmology_eq:Raychaudhuri_equation}) and the generalized
Friedman equation~(\ref{Cosmology_eq:Generalized_Friedmann}) in
Eq.~(\ref{Cosmology_eq:geodesic_deviation_particular}) or, equivalently
in Eq.~(\ref{Cosmology_eq:accel_norm_expansion}), yields
\begin{equation}
\frac{D^{2}\left(n_{\delta}n^{\delta}\right)}{d\tau^{2}}=
\frac{\tilde{D}^{2}\left(n_{\delta}n^{\delta}\right)}{d\tau^{2}}
=\frac{1}{3}\Bigl(8\pi\left(\mu-3p\right)+4\Lambda+
2\mathrm{S}^{\sigma}\mathrm{S}_{\sigma}-{}^{3}R\Bigr)
n_{\delta}n^{\delta}\,,
\label{Cosmology_eq:geodesic_deviation_accel_norm}
\end{equation}
confirming that observers that couple directly with torsion and observers
that do not, will measure the same relative acceleration of the distance
between nearby elements of volume of the fluid.
Although it can also be inferred from
the Raychaudhuri equation~(\ref{Cosmology_eq:Raychaudhuri_equation}),
it is explicit in Eq.~(\ref{Cosmology_eq:geodesic_deviation_accel_norm})
that the square of the norm of the torsion vector field,
$\mathrm{S}^{\sigma}\mathrm{S}_{\sigma}$,
has the same sign of a positive cosmological constant, therefore,
the torsion field also contributes to the positive relative acceleration
of the distance between infinitesimally close test particles, an effect
that is measurable by both type of observers.

\subsection{Dynamics of the cosmic fluid}

In the previous subsection, we have found that the torsion vector
field may contribute to a positive cosmological constant. We are then
interested in understanding if it is possible to have solutions with
zero cosmological constant,  $\Lambda=0$, such that the relative
accelerated expansion measured in our universe is completely fueled by
the torsion field.

To analyze this problem, 
we introduce the following dimensionless parameters
\begin{equation}
\begin{aligned}q & \equiv-1-\frac{\dot{\theta}}{3\mathrm{H}^{2}}\,, &
\Omega_{\mathrm{K}} & \equiv-\frac{^{3}R}{6\mathrm{H}^{2}}\,, &
\Omega & \equiv\frac{8\pi\mu}{3\mathrm{H}^{2}}\,, & \Omega_{\Lambda}
& \equiv\frac{\Lambda}{3\mathrm{H}^{2}}\,, & \Omega_{\mathrm{S}}
& \equiv\frac{\mathrm{S}^{\sigma}\mathrm{S}_{\sigma}}{3\mathrm{H}^{2}}\,,
\end{aligned}
\end{equation}
where $q$ is the acceleration parameter, $\mathrm{H}$ is
the Hubble parameter defined in subsection~\ref{Subsec:Wave-equation},
$\Omega$, $\Omega_{\Lambda}$ and $\Omega_{\mathrm{S}}$ represent,
the matter fields,
dark
energy, and intrinsic spin dimensionless densities,
respectively, where we note that
the density parameter
$\Omega$ accounts for the contribution to the energy density of all
matter fields, be it baryons, photons, dark matter or neutrinos,
and note also that $\Omega_{\mathrm{S}}\geq0$.
We further define the  effective equation
of state parameter $\chi$ as
\begin{equation}
\chi\left(\tau\right)\equiv\frac{p}{\mu}\,.
\label{chi}
\end{equation}
Then, we can rewrite the Raychaudhuri and the Friedman
equations given in Eqs.~(\ref{Cosmology_eq:Raychaudhuri_equation}) and
(\ref{Cosmology_eq:Generalized_Friedmann}) as
$\frac{3}{2}\left(\chi+\frac{1}{3}\right)\Omega-
\Omega_{\Lambda}-2\Omega_{\mathrm{S}}=q\,,$
and  $\Omega_{\mathrm{K}}+\Omega+\Omega_{\Lambda}-\Omega_{\mathrm{S}}  =1$,
 respectively.
The measured empirical results indicate that the universe is very close
to being Ricci flat, hence
one can put here $\Omega_{\mathrm{K}}=0$.
So, the Raychaudhuri and the  Friedman equations turn into
\begin{align}
&\frac{3}{2}\left(\chi+\frac{1}{3}\right)\Omega-
\Omega_{\Lambda}-2\Omega_{\mathrm{S}} =q\,.
\label{Omegaandq}
\\
&\Omega+\Omega_{\Lambda}-\Omega_{\mathrm{S}}  =1\,,
\label{Omegas}
\end{align}
respectively. 
We see that in the Raychaudhuri equation~\eqref{Omegaandq},
 $\Omega_{\mathrm{S}}$ has the same sign of the cosmological constant
 term $\Omega_{\Lambda}$, and thus contributes to the acceleration of
 the expanding universe, as it could be expected, since the spin can
 be thought of as a source of centrifugation for the universe.  On the
 other hand, in the Friedman equation~\eqref{Omegas},
 $\Omega_{\mathrm{S}}$ has a minus sign relative to the cosmological
 constant term $\Omega_{\Lambda}$, as it is also expected, since the
 spin can be thought of as a kinetic term and thus contributes to the
 balance of the kinetic energy of the universe, which is represented
 in the 1 of the right-hand side of the equation.
Let us see more clearly
the effect of $\Omega_{\mathrm{S}}$ by
turning off $\Lambda$.
Setting $\Omega_{\Lambda}=0$ in Eqs.~\eqref{Omegaandq} and \eqref{Omegas}
we find
$\frac{3}{2}\left(\chi+\frac{1}{3}\right)
\Omega-2\Omega_{\mathrm{S}}  =q$
and
$\Omega-\Omega_{\mathrm{S}}  =1$, respectively.
In this
case, the Friedman equation with $\Omega_{\Lambda}=0$
and $\Omega_{\mathrm{S}}\geq0$
necessarily
implies that $\Omega\geq1$.
On the other hand, the 
matter density in our universe is known to be less than one, $\Omega<1$.
Therefore solutions with $\Omega_{\Lambda}=0$ are excluded. 
This fact, namely, that solutions of the considered model where
$\Omega_{\Lambda}=0$ are excluded, of
course, does not mean that the model is excluded.
This simply implies
that if a torsion field of the considered type exists,
$\Omega_{\mathrm{S}}$
cannot
solely
contribute 
to the expansion of the universe described by the
Raychaudhuri equation, given in Eq.~\eqref{Omegaandq},
and to the energy balance of the
Friedman equation given in Eq.~\eqref{Omegas},
$\Omega_{\Lambda}$ has also to contribute.
How much is this $\Omega_{\mathrm{S}}$
contribution  depends on the matter model.
Following \citep{Ray_Smalley_1983,Obukhov_Korotky_1987}
if we consider that the
cosmological spinning fluid is a medium whose elements are galaxies
and galaxy clusters, the torsion tensor would be caused by their
macroscopic angular momenta. In that case, the spin density
$\Omega_{\mathrm{S}}$
could
provide a nonnegligible contribution to the accelerated expansion of
the universe.
On the other hand,  if the only source of the
torsion tensor is the intrinsic spin of elementary particles,
where one finds that 
the
value of $\Omega_{\mathrm{S}}$ is much smaller than $\Omega$,
see \citep{Hehl_Heyde_Kerlick_1974,Kerlick_1973}, its
contribution to the accelerated expansion of the universe would be
negligible.

To conclude the analysis, we remark that independently of the actual
value of the dimensionless spin density $\Omega_{\mathrm{S}}$, if
$\Omega_{\mathrm{S}}$ is nonzero, i.e.,
$\Omega_{\mathrm{S}}>0$, the results in
Theorems~\ref{Theorem:Negative_zero_curvature} and
\ref{Theorem:Global_not_closed} are still verified. So, even though
the contribution of the torsion tensor to the accelerated expansion of
the universe may be negligible, it still markedly changes the geometry
and the allowed topology of the spacetime.

\section{Conclusions\label{Section:Conclusion}}

We presented the general set of structure equations for the 1+3
spacetime decomposition in 4 spacetime dimensions, valid for any
theory of gravitation based on a metric compatible affine connection,
showing in complete generality the relations between the kinematical
quantities of the timelike congruence, the torsion tensor and the Weyl
and the Ricci tensors.

The new equations were then used to study solutions of the
Einstein-Cartan theory with a cosmological perfect fluid having an
intrinsic spin, such that the geometry of the spacetime is described
by both the metric and the torsion tensor field. The model showed that
even in the presence of a torsion field originated by the intrinsic spin
of matter, the metric tensor can be described by a general, spatially
isotropic and homogeneous, FLRW solution. Here we would like to
highlight that although we have assumed that the torsion tensor
has the intrinsic spin of the fluid's constituents as its source, this does
not imply that the fluid's elements of volume must have a nonzero
intrinsic spin density. As was shown in
\citep{Hehl_Heyde_Kerlick_1974}, even if in an element of volume
containing many particles with the intrinsic spins of the individual
particles are randomly oriented, such that the average spin density is
zero, the variance is not zero, hence the torsion tensor is not
zero. This, in our view, is the correct approach to the
Einstein-Cartan theory, where the fluid is described by a
semi-classical model, whose elements of volume contain many particles.
Of course, it could also be the case that the individual spins are
aligned, and the average intrinsic spin density is not zero. In either
case, in the considered model the torsion tensor is not zero.
Although the metric tensor was found to be described by a FLRW model,
it was shown that the Weyl tensor might not vanish, which leads to
very strong constrains on the allowed geometry and topology of the
spacetime. Indeed, due to the coupling between the torsion and Weyl
tensors, in the considered model, the universe must either be flat or
open.

In the open case, we then derived a wave equation for the traceless
part of the magnetic part of the Weyl tensor, concluding that the
presence of intrinsic spin of matter may induce gravitational waves,
providing, to our knowledge, the first explicit result showing that
the torsion field may source or influence the emission of
gravitational waves in a cosmological setting. Although these waves
are strongly damped in an
expanding universe, this result may provide a smoking gun for the
presence of spacetime torsion.

In the considered model, it was also possible to determine that
a torsion tensor field originated from intrinsic spin
contributes to the
positive accelerated expansion of the universe, nonetheless, comparing the
theoretical predictions of the model with the current experimental
data, the torsion tensor cannot completely replace the role of the
cosmological constant.

\begin{acknowledgments}
We acknowledge Funda\c c\~ao para a Ci\^encia e Tecnologia - FCT, Portugal,
for financial support through Project No.~UIDB/00099/2020.
\end{acknowledgments}

\appendix

\section{\label{Section:Appendix}Properties of the
Laplace-Beltrami harmonics}
\label{laplacebeltramiharmonics}

\subsection{Scalar harmonics}

In this appendix we list some of the properties of the scalar, vector
and tensor eigenfunctions of the covariantly defined Laplace-Beltrami
operator on 3-hypersurfaces of constant curvature used to define the so
called scalar-vector-tensor decomposition. For concreteness,
we consider a spacetime endowed with a FLRW metric, such that the
homogeneous spatial sections represent such 3-hypersurfaces. 

Let $Q^{k}$ represent the scalar eigenfunctions of the covariantly
defined Laplace-Beltrami operator
$\tilde{D}^{2}\equiv\tilde{D}_{\alpha}\tilde{D}^{\alpha}$,
where $\tilde{D}$ represents the covariant derivative operator associated
with the Levi-Civita connection, such that
\begin{equation}
\tilde{D}^{2}Q^{k}=-
\frac{k^{2}}{\ell^{2}}Q^{k}\,,
\end{equation}
and $\dot{Q}^{k}=0$, where $\ell$ represents the scale factor defined
in Eq.~(\ref{Cosmology_eq:scale_factor_expansion_definition}) and the
harmonic index $k$ may take discrete or continuous values depending on
whether $K=+1$, or $K\in\left\{ -1,0\right\} $, respectively, where $K$
was introduced in Eq.~\eqref{Cosmology_eq:R3_constant_general}.  Then, we can
define the following tensors from the scalar eigenfunctions $Q^{k}$,
\begin{equation}
\begin{aligned}Q^{\left(0\right),k} & _{\alpha}=
-\frac{\ell}{k}\tilde{D}_{\alpha}Q^{k}\,,\\
Q^{\left(0\right),k}{}_{\alpha\beta} & =
\frac{\ell^{2}}{k^{2}}\tilde{D}_{\beta}\tilde{D}_{\alpha}Q^{k}+
\frac{1}{3}h_{\alpha\beta}Q^{k}\,,
\end{aligned}
\label{Appendix_eq:Vector_tensor_scalar_harmonics}
\end{equation}
with the following properties
\begin{equation}
\begin{aligned}\dot{Q}^{\left(0\right),k}{}_{\alpha} & =0\,,\\
\tilde{D}^{\mu}Q^{\left(0\right),k}{}_{\mu} & =
\frac{k}{\ell}Q^{k}\,,\\
\tilde{D}^{2}Q^{\left(0\right),k}{}_{\alpha} & =
\frac{2K-k^{2}}{\ell^{2}}Q^{\left(0\right),k}{}_{\alpha}\,,\\
\tilde{D}_{[\alpha}
\tilde{D}_{\beta]}Q^{\left(0\right),k}{}_{\gamma} &
=\frac{K}{2\ell^{2}}\left(h_{\alpha\gamma}
Q^{\left(0\right),k}{}_{\beta}-h_{\beta\gamma}
Q^{\left(0\right),k}{}_{\alpha}\right)\,,
\\
Q^{\left(0\right),k}{}_{\mu}{}^{\mu} & =0\,,\\
\dot{Q}^{\left(0\right),k}{}_{\alpha\beta} & =0\,,\\
\varepsilon_{\left(\alpha\right|}{}^{\mu\nu}D_{\nu}
Q^{\left(0\right),k}{}_{\mu\left|\beta\right)} & =0\,,\\
\tilde{D}^{\mu}Q^{\left(0\right),k}{}_{\alpha\mu} &
=\frac{2}{3\ell k}\left(k^{2}-3K\right)
Q^{\left(0\right),k}{}_{\alpha}\,,\\
\tilde{D}^{2}Q^{\left(0\right),k}{}_{\alpha\beta} &
=\frac{6K-k^{2}}{\ell^{2}}Q^{\left(0\right),k}{}_{\alpha\beta}\,,
\end{aligned}
\end{equation}
where in the previous expressions the $k^{-1}$ factor
is not a problem
because $\tilde{D}_{\alpha}Q^{0}=0$.

\subsection{Vector harmonics}

Given a sufficiently smooth 1-form field $Y_{\alpha}$ in a FLRW
spacetime, we can in general decompose it as
\begin{equation}
Y_{\alpha}=
\sum_{k}\mathrm{T}_{k}^{\left(0\right)}Q^{\left(0\right),k}{}_{\alpha}
+\mathrm{T}_{k}^{\left(1\right)}Q_{\alpha}^{\left(1\right),k}\,,
\end{equation}
where the vectors $Q_{\alpha}^{\left(0\right),k}$
are obtained from
the scalar eigenfunctions $Q^{k}$,
Eq.~(\ref{Appendix_eq:Vector_tensor_scalar_harmonics}),
and $Q^{\left(1\right),k}{}_{\alpha}$ represent the solutions of
the vector Helmholtz equation
\begin{equation}
\tilde{D}^{2}Q^{\left(1\right),k}{}_{\alpha}=-
\frac{k^{2}}{\ell^{2}}Q^{\left(1\right),k}{}_{\alpha}\,,
\end{equation}
with the following properties
\begin{equation}
\begin{aligned}\dot{Q}^{\left(1\right),k}{}_{\alpha} & =0\,,\\
\tilde{D}^{\mu}Q^{\left(1\right),k}{}_{\mu} & =0\,.
\end{aligned}
\end{equation}
Similarly to the scalar eigeinfunctions $Q^{k}$, we can find a set
of 2-tensors associated with $Q^{\left(1\right),k}{}_{\alpha}$:
\begin{equation}
Q^{\left(1\right),k}{}_{\alpha\beta}=-
\frac{\ell}{2k}\left(\tilde{D}_{\alpha}
Q^{\left(1\right),k}{}_{\beta}+\tilde{D}_{\beta}
Q^{\left(1\right),k}{}_{\alpha}\right)\,,
\end{equation}
with the following properties
\begin{align}
Q^{\left(1\right),k}{}_{\mu}{}^{\mu} & =0\,,\nonumber \\
\dot{Q}^{\left(1\right),k}{}_{\alpha\beta} & =0\,,\nonumber \\
\tilde{D}^{\mu}Q^{\left(1\right),k}{}_{\alpha\mu} &
=\frac{k^{2}-2K}{2\ell k}Q^{\left(1\right),k}{}_{\alpha}\,,\\
\tilde{D}^{2}Q^{\left(1\right),k}{}_{\alpha\beta} &
=\frac{4K-k^{2}}{\ell^{2}}Q^{\left(1\right),k}{}_{\alpha\beta}
\,.\nonumber 
\end{align}

\subsection{Tensor harmonics}

Given a general smooth 2-tensor field $Y_{\alpha\beta}$ in a FLRW
spacetime, we can in general decompose it as
\begin{equation}
Y_{\alpha\beta}=\sum_{k}\mathrm{T}_{k}^{\left(0\right)}
Q^{\left(0\right),k}{}_{\alpha\beta}+\mathrm{T}_{k}^{\left(1\right)}
Q^{\left(1\right),k}{}_{\alpha\beta}+\mathrm{T}_{k}^{\left(2\right)}
Q^{\left(2\right),k}{}_{\alpha\beta}\,,
\end{equation}
where the 2-tensors $Q^{\left(2\right),k}{}_{\alpha\beta}$
are defined
as the solutions of the tensor Helmholtz equation
\begin{equation}
\tilde{D}^{2}Q^{\left(2\right),k}{}_{\alpha\beta}=-
\frac{k^{2}}{\ell^{2}}Q^{\left(2\right),k}{}_{\alpha\beta}\,.
\end{equation}
These verify
\begin{align}
Q^{\left(2\right),k}{}_{\mu}{}^{\mu} & =0\,,\nonumber \\
\dot{Q}^{\left(2\right),k}{}_{\alpha\beta} & =0\,,\\
\tilde{D}^{\mu}Q^{\left(2\right),k}{}_{\alpha\mu} & =0\,.
\nonumber 
\end{align}

\end{document}